\documentclass[journal]{IEEEtran}
\usepackage{amsmath}
\usepackage{graphicx}
\usepackage{epstopdf}
\usepackage{subfigure}
\usepackage{amsmath}
\usepackage{algorithm}
\usepackage{algpseudocode}
\usepackage{cite}
\usepackage{multicol}
\usepackage{multirow}
\usepackage{mathtools}
\usepackage{amssymb}
\usepackage{amsthm}
\usepackage{color}
\usepackage{bm}
\usepackage{indentfirst}
\usepackage{url}
\usepackage{nomencl}
\usepackage{booktabs}
\usepackage{cite}
\usepackage[numbers,sort&compress]{natbib}
\usepackage{amsmath}
\usepackage{graphicx}
\usepackage{subfigure}
\usepackage{amsmath}
\usepackage{algorithm}
\usepackage{algpseudocode}
\usepackage{cite}
\definecolor{mygray}{gray}{.9}
\usepackage{multicol}
\usepackage{multirow}
\usepackage{mathtools}
\usepackage{amssymb}
\usepackage{color}
\usepackage{bm}
\usepackage{indentfirst}
\usepackage{url}
\usepackage{nomencl}
\usepackage{booktabs}
\usepackage{cite}
\usepackage[numbers,sort&compress]{natbib}

\newtheorem{Theorem}{Theorem}
\newtheorem{Lemma}{Lemma}

\newtheorem{Corollary}{Corollary}
\newtheorem{Definition}{Definition}
\newtheorem{Remark}{Remark}
\newtheorem{Assumption}{Assumption}
\allowdisplaybreaks
\begin{document}
\title{Social Profit Optimization with Demand Response Management in Electricity Market: A Multi-timescale Leader-following Approach}

\author{Jianzheng Wang, Yipeng Pang,
	Guoqiang~Hu,~\IEEEmembership{Senior Member,~IEEE} % <-this % stops a space
	\thanks{This work was supported in part by Singapore Economic Development Board under EIRP grant S14-1172-NRF EIRP-IHL, and in part by the Republic of Singapore s National Research Foundation under its Campus for Research Excellence and Technological Enterprise (CREATE) Programme through a grant to the Berkeley Education Alliance for Research in Singapore (BEARS) for the Singapore-Berkeley Building Efficiency and Sustainability in the Tropics (SinBerBEST) Program.}% <-this % stops a space
	\thanks{Jianzheng Wang, Yipeng Pang and Guoqiang Hu are with the School of Electrical and Electronic Engineering, Nanyang Technological University, Singapore, 639798 e-mail: (wang1151@e.ntu.edu.sg, ypang005@e.ntu.edu.sg, gqhu@ntu.edu.sg).}
% <-this % stops a space
}

\maketitle

\begin{abstract}
In the electricity market, it is quite common that the market participants make ``selfish'' strategies to harvest the maximum profits for themselves, which may cause the social benefit loss and impair the sustainability of the market in the long term. Regarding this issue, we will study how the social profit can be improved through strategic demand response management. Specifically, we explore two interaction mechanisms in the market: Nash game and Stackelberg game. At the user side, each user makes the respective energy-purchasing strategy to optimize its own profit. At the utility company side, we consider multiple self-centric utility companies who play games. A social-centric governmental utility company is established as the leader to optimize the social profit of the market through competitions. Then, a multi-timescale leader-following problem of the utility companies is formulated under the coordination of an independent system operator. By our proposed demand function amelioration strategy, the market efficiency is maximized. In addition, by considering some additional constraints of the market, two projection-based algorithms are proposed. The feasibility of the proposed algorithms is verified with an IEEE 9-bus system model in the simulation.
\end{abstract}


\IEEEpeerreviewmaketitle
\section{Introduction}
%\IEEEPARstart
\subsection{Background and Motivation}
\IEEEPARstart{T}{he} report of ``Grid 2030'' reveals that there are around 119-188 billion dollars annual costs on other industrial areas due to power disturbances and power quality issues \cite{1}. To maintain the reliability of the power system, demand response (DR) management has been proved to be a successful solution and is widely implemented around the world \cite{2}.

Among the various DR research works, social-centric and self-centric profit optimization problems have drawn much attention in the recent few years. Social profit may be impaired in case each agent optimizes its own profit without considering the benefit of the society. To model such a ``selfish'' manner, Nash game problems are usually discussed. Alternatively, if we regard the whole market as a ``cooperative system'', the agents will cooperatively make strategies to optimize the social profit. This motivates us to investigate how to reduce the gap between the solutions to social-centric and self-centric problems such that the Nash equilibrium (NE) can be ``driven'' to the social optimal solution.
To this end, in this work, we focus on an electricity market model with multiple utility companies (UCs) and users, where each agent can make the respective strategy to optimize its own profit, and then develop some effective coordination strategies to improve the market efficiency.

\subsection{Literature Review}

Recent studies on DR can be categorized into two main areas: self-centric optimizations and social-centric optimizations, which can be further classified into unilateral, bilateral and multilateral optimizations depending on the roles of the participants involved \cite{ma2019energy}. Self-centric optimization problems cover a wide range of objectives, such as cost of UCs \cite{bahrami2015demand}, payment of users \cite{jacquot2018analysis}, payoff of aggregators \cite{parvania2013optimal}, peak load reduction \cite{pedrasa2009scheduling}, and recovery of investments \cite{nash}.
However, social-centric optimization problems usually consider the benefit of the whole market, where the objective functions can be the combinations or tradeoffs of those of all the participants. For example, the authors of \cite{knudsen2015dynamic} considered the cost characteristics of a batch of loads and generators. The optimal power flow is settled by solving a constrained social cost minimization problem at different market-clearing instants.
In \cite{c22}, the social cost function was designed as the total cost of load reduction scheduling, load shifting scheduling, energy storage devices, and on-set generators.
A real-time pricing algorithm was designed for an independent system operator (ISO) to optimize the social profit in \cite{26+1}. By formulating the dual problem, the optimal energy prices are settled as the optimal dual variables. In \cite{11}, DR was defined as the difference between the energy intended to buy and actually consumed. The optimal DR price is settled when the overall benefit of DR buyers and DR sellers is maximized.

Even though the existing works on social profit optimizations are fruitful, most of them require that the participants are cooperative or centrally controlled by some coordinators \cite{knudsen2015dynamic,c22,26+1,11}. The discussions on social profit optimizations with multilateral competitions among heterogeneous market participants (e.g., a mixed combination of UCs and users) are still limited. Compared with cooperative or centralized optimizations, the multilateral competition emphasizes the selfish instinct of the agents with heterogeneous objectives. To model the competitive behaviour of agents, game theoretic methods have been widely studied recently \cite{ye2016game,myerson2013game,masoumzadeh2016long,wang2015game,wang2019noncooperative,zhang2019event}. In game problems, each agent makes the best strategy for itself by observing the strategy of its rivals. Therefore, it is of great significance to establish certain coordinator to influence the competition result such that the social profit is optimized \cite{chen2017operating}.
With this motivation, there are two questions that arise here: (i) what's the equilibrium (if exists) of the multilateral competitions in the market? (ii) how to make an effective coordination strategy to optimize the social profit?

To address the above two questions, the concept of market efficiency is employed by this work, which can characterize the gap between the optimized social profit and the ideal maximum social profit. Existing studies on market (system) efficiency can be found in various fields. For instance, a Stackelberg routing problem was discussed in \cite{korilis1997achieving}, where the maximal efficiency can be achieved when the demand of the leader is higher than certain threshold. The author of \cite{roughgarden2004stackelberg} studied a job scheduling problem with a set of machines to minimize the total latency of the system. It shows that the total latency is $1/\epsilon$ times that of the centralized scheduling scenario, where $\epsilon$ is the proportion of the leader machines.
A decentralized electric vehicle scheduling problem was considered in \cite{chakraborty2013flexible}. It shows that if the price of scheduling is proportional to the energy consumption of users, the social benefit loss of the users is 25$\%$ at most compared with the centralized optimization scenario.

The contributions of this work are summarized as follows.
\begin{enumerate}
  \item We propose a multi-UC-multi-user electricity market model based on NE and Stackelberg equilibrium (SE) analysis. At the user side, each user makes the optimal energy-purchasing strategy for itself. At the UC side, some self-centric UCs play games to optimize their own profits. The optimal solutions of users and UCs are analytically derived.
  \item To optimize the social profit, we first propose a basic multi-timescale leader-following optimization scheme by considering a social-centric governmental UC (g-UC, leader) in the market model. Then, to maximize the market efficiency, a demand function amelioration (DFA) strategy is proposed and a multi-timescale leader-following problem with varying leader-following sensitivities is formulated. In addition, we propose two projected updating algorithms for the UCs by considering some additional constraints of the market.
  \item Compared with the existing works on market (system) efficiency, the proposed optimization strategies have the following features. {{(i)}} Different from the leader-following approaches discussed in \cite{korilis1997achieving,roughgarden2004stackelberg}, only one leader is established, i.e., we do not enforce the leader's influence by increasing the proportion/number of the leaders. {{(ii)}} Different from financial incentive based strategies as studied in \cite{chakraborty2013flexible}, we do not set incentive budget for the coordinator. {{(iii)}} By the DFA strategy, the UCs do not know the real demand function of users, which overcomes the privacy releasing issue of conventional Stackelberg game approaches.
\end{enumerate}

The rest of this paper is organized as follows. In Section II, we introduce a Nash-Stackelberg game based electricity market model. The optimal strategies of users and UCs are analytically derived. Section III presents our proposed social profit optimization algorithms based on multi-timescale leader-following approach. In addition, two projected updating algorithms are provided by considering additional constraints of the market. The effectiveness of the proposed algorithms is verified by using an IEEE 9-bus system model in Section IV. Section V concludes this paper.

%--------------------------------------------------------------------------

%--------------------------------------------------------------------------

\section{Problem Formulation}

{In this section, we first present the mathematical model of users and UCs. Then, a Nash-Stackelberg game is formulated for the market.}

We consider an electricity market model composed of user set $\mathcal{M} = \{ 1,2,...,M\} $
%sub-group (SG) set $\mathcal{R}= (1,...ik,...MN)$,
and UC set $\mathcal{N} = \{1,2,...,N\}$, $M \geq 2$, $N \geq 2$.
%SG $ik$ is the $k$-th group member of user group $i$ who trades with UC $k$. So the realization of optimal profit of user $i$ requires the cooperation of SGs.
Let $d_{i,k}$ be the energy-purchasing strategy of user $i$ with UC $k$, and let $p_k$ be the pricing strategy of UC $k$.
In this market, users send their demands to the ISO. Then, the ISO clears the market with UCs based on some coordination strategies.

\subsection{User Side Modeling}
The profit function of user $i$ is defined as the total benefit by consuming energy $d_{i,k}$, which is
\begin{align}\label{21s}
& U_{i} (\mathbf{d}_{i}) = \nonumber \\
& \left\{\begin{array}{ll}
              \sum_{k \in \mathcal{N}} (\alpha d_{i,k}-\dfrac{\beta}{2} d^2_{i,k}-p_k d_{i,k} ), & 0 \leq d_{i,k} \leq \dfrac{\alpha}{\beta}, \\
              \sum_{k \in \mathcal{N}} (\dfrac{\alpha^2}{2\beta}-\dfrac{\alpha p_k}{\beta} ), &   d_{i,k} > \dfrac{\alpha}{\beta},
            \end{array}\right.
\end{align}
where $\mathbf{d}_i=(d_{i,1},...,d_{i,N})^{\top}$ and $\alpha,\beta>0$. In the brackets of the upper formula in (\ref{21s}), the first two terms characterize the benefit by consuming energy $d_{i,k}$ \cite{26+1} and the third term is the payment to UC $k$. The lower formula implies that the benefit of users will be constant when the consumed energy is over certain threshold. Let $y_{i}$ be the total energy demand of user $i$. Assume that $y_i  \leq \alpha/\beta$, then we only need to consider the upper formula in (\ref{21s}).
In this case, the marginal profit of user $i$ decreases with $d_{i,k}$ increasing, which is consistent with the characteristics of many benefit functions \cite{26+2}.

{\begin{Remark}
In this work, we consider a fixed total demand $y_i$ for user $i$ during certain optimization interval, which is also studied in many existing works \cite{bai2017distributed,vlachos2013demand}. By contrast, as we will see later, the optimal demand satisfied by different UCs is elastic and responsive to different price levels. This settlement realizes an optimal tradeoff among different UCs and reflects a smart decision-making process of users.
\end{Remark}}
Then the profit optimization problem of user $i$ can be formulated as
\begin{align}\label{con}
{\mathbf{(P1)}}  \quad \max \limits \quad & U_{i} (\mathbf{d}_i) \nonumber \\
\hbox{subject to} \quad & \sum_{k \in \mathcal{N}} d_{i,k} = y_i.
\end{align}

\begin{Theorem}\label{conx}
The optimal solution to {Problem (P1)} is $\mathbf{d}^*_{i} = ({d}^*_{i,1},...,{d}^*_{i,N})^{\top}$, where
\begin{equation}\label{con1}
d^*_{i,k}=\frac{1}{\beta}(\frac{1}{N}\sum_{j \in \mathcal{N}}p_j-p_k )+\frac{y_{i}}{N}, \quad i \in \mathcal{M}, k \in \mathcal{N}.
\end{equation}
\end{Theorem}

\begin{proof} See {Appendix \ref{conx1}}.
\end{proof}

\begin{Corollary}
Based on the optimal strategy $\mathbf{d}^*_{i}$, the profit of user $i$ can be obtained by (\ref{21s}), which is
\begin{align}\label{u1}
  U^*_{i} (\mathbf{p}) = & \sum_{k\in \mathcal{N}}((\alpha - p_k)(\frac{1}{N\beta}\sum_{j \in \mathcal{N}}p_j-\frac{p_k}{\beta}+\frac{y_{i}}{N}) \nonumber \\
  & -\frac{\beta}{2}(\frac{1}{N\beta}\sum_{j \in \mathcal{N}}p_j-\frac{p_k}{\beta}+\frac{y_{i}}{N})^2).
\end{align}
The sales quantity of UC $k$ is
\begin{equation}\label{con2}
d_{k} = \sum_{i \in \mathcal{M}} d^*_{i,k} = \frac{M}{\beta}(\frac{1}{N}\sum_{j \in \mathcal{N}}p_j-p_k )+\frac{Y}{N},
\end{equation}
where $Y=\sum_{i \in \mathcal{M}}y_i$.
\end{Corollary}

{\begin{Remark}\label{323}
In (\ref{21s}), the total profit of user $i$ consists of $N$ componential profit functions. Similar models can be found in \cite{21,11,wang2019social}. This model breaks the conventional impression that users only purchase the energy from the UC who provides the lowest price, which conflicts with many multi-price markets. In the real markets, the users may acquire energy from different UCs with different prices due to various reasons, e.g., price, quality of service, reputation, etc. Therefore, this model allows the UCs with different prices to co-exist and each of them can harvest certain benefit. As explained by (\ref{con1}), the energy quantity purchased from UC $k$ will be lower with a relatively higher price since the sign of $p_k$ is negative, and vice versa.
\end{Remark}}

\subsection{Utility Company Side Modeling}\label{ac}

The profit function of UC $k$ is designed as
\begin{equation}\label{gge}
W_k (p_k,d_k)=p_k d_{k}-D_k(d_k),
\end{equation}
where $p_k d_{k}$ represents the revenue of sales quantity $d_k$ and $D_k(d_k) = a_k d^2_k +b_k d_k +c_k$ is the generation cost, $a_k>0, b_k\geq 0, c_k\geq 0$.
By plugging (\ref{con2}) into (\ref{gge}), the profit function of UC $k$ can be obtained as
\begin{align}\label{81}
W_k (p_k,\mathbf{p}_{-k})=A_k p^2_k+B_k p_k+C_k,
\end{align}
where
\begin{subequations}\label{e2}
\begin{align}
A_k = & -\dfrac{a_kM^2}{N^2}(\dfrac{N-1}{\beta})^2-\dfrac{M(N-1)}{N\beta}, \label{7a} \\
B_k= & \dfrac{2a_kM(N-1)}{N^2\beta}(\dfrac{M\sum_{j \in \mathcal{N}\setminus{\{k\}}}p_j}{\beta}+Y ) \nonumber \\
&+\dfrac{b_kM(N-1)}{N\beta}+\dfrac{1}{N}(\dfrac{M\sum_{j \in \mathcal{N}\setminus{\{k\}}}p_j}{\beta}+Y), \label{7b} \\
C_k= & -\dfrac{a_k}{N^2}(\dfrac{M\sum_{j \in \mathcal{N}\setminus{\{k\}}}p_j}{\beta}+Y)^2 \nonumber \\
& -\dfrac{b_k}{N}(\dfrac{M\sum_{j \in \mathcal{N}\setminus{\{k\}}}p_j}{\beta}+Y)-c_k,
\end{align}
\end{subequations}
and $\mathbf{p}_{-k}=(p_1,...,p_{k-1},p_{k+1},...,p_N)^{\top}$. Then, the optimization problem of UCs can be given by
\begin{align}\label{gene}
{\mathbf{(P2)}}  \quad &\max\limits_{p_k}  \quad W_k({p_k},\mathbf{p}_{-k}), \quad \forall k \in \mathcal{N}.
\end{align}
\begin{Remark}\label{r1}
In {Problem (P2)}, we tentatively do not set any lower bound for $p_k$ since the market may encourage the users to consume energy by using negative prices with some economic purposes, which has been applied in some practical electricity markets \cite{fanone2013case}. The optimization problems with some additional constraints will be discussed in Section \ref{c}.
\end{Remark}
To illuminate the optimization problem of UCs and users, we employ the following two definitions.
\begin{Definition}\label{sta2}({Stackelberg Equilibrium }\cite{bacsar1998dynamic})
Let UC $k \in \mathcal{N}$ be the Stackelberg leader with objective function $W_k(p_k,{d}_{k}(p_k,\mathbf{p}_{-k}))$ and the users be the followers. Strategy $p^{\mathrm{SE}}_k$ is named as an SE for UC $k$ if
\begin{equation}\label{}
\begin{split}
W_k(p^{\mathrm{SE}}_k, {d}_{k}(p^{\mathrm{SE}}_k,\mathbf{p}_{-k}))
=&\max \limits_{p_k}  W_k(p_k, {d}_{k}(p_k,\mathbf{p}_{-k})). \nonumber
\end{split}
\end{equation}
\end{Definition}

\begin{Definition}\label{nas}({Nash Equilibrium }\cite{bacsar1998dynamic})
$\mathbf{p}^{\mathrm{NE}}=({p}_1^{\mathrm{NE}},...,{p}_N^{\mathrm{NE}})^{\top}$ is named as an NE of UCs if $W_k(p_k,\mathbf{p}^{\mathrm{NE}}_{-k}) \leq W_k(p^{\mathrm{NE}}_k,\mathbf{p}^{\mathrm{NE}}_{-k})$ holds, $\forall k \in \mathcal{N}$.
\end{Definition}

As seen from (\ref{gge}) and (\ref{81}), the objective function of UC $k$ contains the pricing strategy of other UCs and the optimal strategy of users, which complies with Definitions \ref{sta2} and \ref{nas}. In {Problem (P2)}, $W_k(p_k,\mathbf{p}_{-k})$ is concave at $p_k$. Hence, the optimal strategy of UC $k$ can be obtained by solving the first-order optimality condition $\nabla_{p_k} W_k(p_k,\mathbf{p}_{-k}) =0$, which gives
\begin{equation}\label{332}
2A_k p_k + B_k = 0.
\end{equation}
By writing (\ref{332}) in a compact form over set $\mathcal{N}$, we have
\begin{align} \label{lp3}
&  \mathbf{A}\mathbf{p}=\mathbf{q},
\end{align}
where
\begin{subequations}
\begin{align}
& \mathbf{A}=[s_{kk'}]_{k,k' \in \mathcal{N}},\\
& s_{kk}=2a_kM^2(N-1)^2+2MN\beta (N-1),  \label{lp66} \\
& s_{kj}=-(2a_kM^2(N-1)+MN\beta),  \label{lp67}\\
& \mathbf{p}=(p,...,p_k,...,p_N)^{\top}, \\
& \mathbf{q}=(q_1,...,q_k,...,q_N)^{\top}, \\
&q_k=\beta Y(N\beta+2a_kM(N-1))+N\beta M b_k (N-1).
\end{align}
\end{subequations}
By (\ref{lp66}) and (\ref{lp67}), it can be verified that $\mid s_{kk} \mid -(N-1)\mid s_{kj} \mid =MN\beta (N-1)>0$, which means $\mathbf{A}$ is strictly diagonally dominant and invertible \cite[Thm. 6.1.10]{ra}. Therefore, the solution to (\ref{lp3}) (i.e., NE of UCs) exists and is unique, which is
\begin{equation}\label{16}
  \mathbf{p}^{\mathrm{NE}}=\mathbf{A}^{-1}\mathbf{q}.
\end{equation}

\begin{Corollary}\label{nb}
({Rationality of $\mathbf{d}^*_i$}) Given the $\mathbf{p}^{\mathrm{NE}}$, $\mathbf{d}^*_{i}$ solved by (\ref{con1}) is said rational if $\mathbf{d}^*_{i} \geq \mathbf{0}$. A closed-loop condition of the rationality of $\mathbf{d}^*_{i}$ is
\begin{equation}\label{14}
\beta \mathbf{y}_i + \mathbf{B}\mathbf{A}^{-1}\mathbf{q} \geq \mathbf{0},
\end{equation}
where $\mathbf{B} = - N \mathbf{I}_N + \mathbf{1}_N \mathbf{1}_N ^{\top}$ and $\mathbf{y}_i=(y_i,...,y_i)^{\top}\in \mathbb{R}^N$.
\end{Corollary}

\begin{proof}
See {Appendix \ref{c3p}}.
\end{proof}

{\begin{Remark}
The condition provided by Corollary 2 ensures a positive optimal energy demand $d^*_{i,k}$ in (\ref{con1}). (\ref{14}) can be satisfied if demand $y_i$ is larger than or equal to the largest element in vector $-\mathbf{B}\mathbf{A}^{-1}\mathbf{q}/ \beta$. By contrast, $d^*_{i,k}<0$ means user $i$ sales energy of quantity $|d^*_{i,k}|$ to UC $k$ with price $p_k$, which is out of the scope of our discussion. A general solution for non-negative energy demands will be discussed in Section \ref{c}.
\end{Remark}}
Based on the previous discussion, the proposed Nash-Stackelberg game based energy trading mechanism is illustrated in Fig. \ref{11}. The optimization procedure of UCs and users is summarized in {Algorithm 1}.
\begin{figure}
  \centering
  \includegraphics[width=10cm]{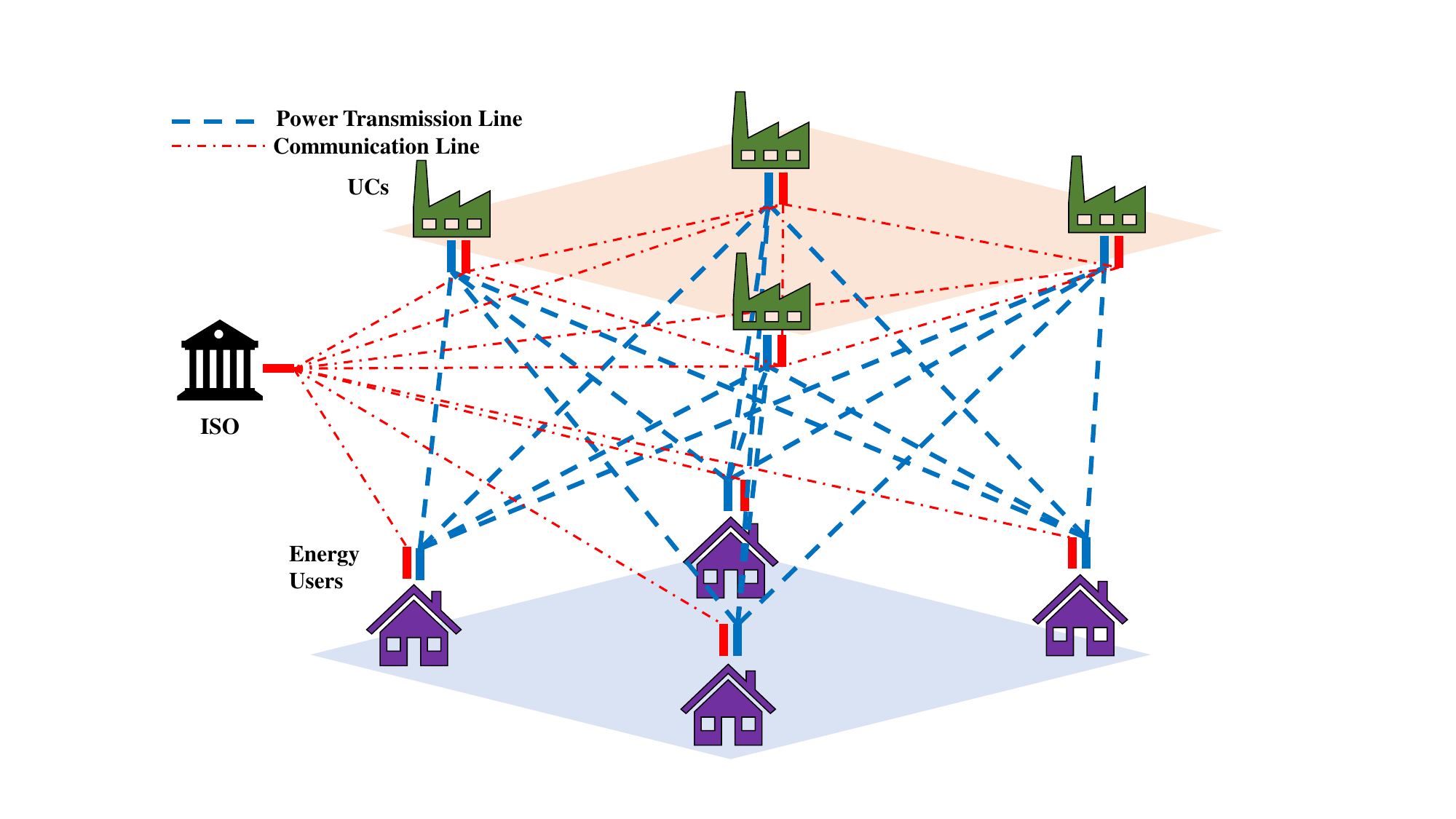}\\
  \caption{Energy trading mechanism of the proposed market by Nash-Stackelberg game.} \label{11}
\end{figure}

\begin{algorithm}[H]
\caption{Nash-Stackelberg game among UCs and users}
\begin{algorithmic}[1]
\State Determine $M$, $N$, $\alpha$, $\beta$, $y_i$, $a_k$, $b_k$, $c_k$, $i \in \mathcal{M}$, $k \in \mathcal{N}$;
\State Solve Nash equilibrium $\mathbf{p}^{\mathrm{NE}}$ by (\ref{16});
\State Solve the optimal strategy of users by (\ref{con1}) with $\mathbf{p}^{\mathrm{NE}}$ obtained in Step 2.
\end{algorithmic}
\end{algorithm}

\subsection{Social Profit Function}

In the proposed model, if we view all the participants as an aggregation, the social profit function can be defined by $S(\mathbf{p})=\sum_{i \in \mathcal{M}}U^*_{i}(\mathbf{p})+\sum_{k \in \mathcal{N}}W_k(\mathbf{p})$. Then, the social profit optimization problem can be formulated as\footnote[1]{In Problem (P3), we do not consider the constraint on $\mathbf{p}$ tentatively, as explained in Remark \ref{r1}.}
\begin{align}\label{xww}
{\mathbf{(P3)}} \quad &\max\limits_{\mathbf{p}} \quad  S (\mathbf{p}).
\end{align}
However, in some practical cases as discussed in the Introduction section, UCs may choose some self-centric optimal strategies, e.g., $\mathbf{p}^{\mathrm{NE}}$, rather than the optimal solution to {Problem (P3)}. Hereby, to characterize the market efficiency, we utilize the concept of price of anarchy (PoA), which is defined by \cite{poaa}
\begin{equation}\label{poa}
    \mathrm{PoA}= \sup_{\mathbf{p}^{\mathrm{NE}}} \frac{S(\mathbf{p}^{\diamond})}{S(\mathbf{p}^{\mathrm{NE}})},
\end{equation}
where $\mathbf{p}^{\diamond}$ denotes the optimal solution to {Problem (P3)}. Therefore, PoA characterizes the largest profit loss of all possible social profits at NE compared with the ideal maximum social profit. {By (\ref{16}), the NE is unique, which means the market efficiency loss can be characterized by the gap between the solutions to {Problems (P2)} and {(P3)}.}

\begin{Assumption}\label{as1}
Assume that $2(a_k+a_j)M(N-1)+N\beta-2 M\sum_{h \in \mathcal{N} \setminus \{j,k\}} a_h>0 $, $\forall k,j \in \mathcal{N}$, $k \neq j$.
\end{Assumption}
Since $N \geq 2$, {Assumption \ref{as1}} holds if $a_k$ is sufficiently close to each other or $N\beta$ is large enough, $k \in \mathcal{N}$.

\begin{Lemma}\label{ov}
Suppose that {Assumption \ref{as1}} holds. Then, $S(\mathbf{p})$ is concave.
\end{Lemma}
\begin{proof} See {Appendix \ref{ov1}}.
\end{proof}

{\begin{Remark}
As seen in Lemma \ref{ov}, Assumption \ref{as1} is used to ensure a concave social profit function, as often discussed in social profit optimization problems \cite{knudsen2015dynamic,samadi2010optimal}.
\end{Remark}}

\section{Social Profit Optimization Strategy}\label{spm}

In this section, we will optimize the social profit of the market by two methods: basic leader-following method and DFA based leader-following method, and further discuss two projection based algorithms by considering additional constraints of the market.

\begin{figure}
  \centering
  \includegraphics[width=8cm]{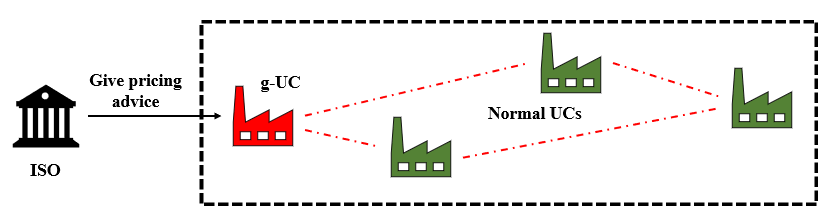}\\
  {\caption{Leader-based UC community with the coordination of ISO (the user level is removed for simplicity).}\label{12}}
\end{figure}

\subsection{Basic Multi-timescale Leader-following Framework}\label{eb}

{To optimize the social profit, we propose a social-centric g-UC (i.e., leader) in the market with the rest UCs being normal ones (i.e., followers). The examples of social-centric g-UCs can be referred to in \cite{agudelo2020drivers}. In this model, the pricing strategy of g-UC can be assigned by the ISO. The proposed leader-based UC community is illustrated in Fig. \ref{12}.

\subsubsection{Leader-following Problem} Define the pricing strategy of g-UC and other normal UCs as $p_1$ (i.e., the index of g-UC is $1$) and $p_k$, respectively, $k \in \widetilde{\mathcal{N}} \triangleq \mathcal{N} \setminus \{1\}$. Note that in NE scenario, the strategy of UC $k \in \widetilde{\mathcal{N}}$ is influenced by the strategy of UC $j \in \mathcal{N}\setminus \{k\}$. In the following, we will characterize the NE and SE of normal UCs with certain given $p_1$ by using the following two definitions (with a slight abuse of notation, we still use $p^{\mathrm{SE}}_k$ and $p^{\mathrm{NE}}_k$ to represent SE and NE in the new games, respectively).
\begin{Definition}\label{sta2-1}
(Leader-based SE) Let UC $k \in \widetilde{\mathcal{N}}$ be the Stackelberg leader with objective function $W_k(p_k,{d}_{k}(p_k,\mathbf{p}_{-k}))$ and the users be the followers. Strategy $p^{\mathrm{SE}}_k$ is named as an SE of normal UC $k \in \widetilde{\mathcal{N}}$ if
\begin{equation}\label{}
\begin{split}
W_k(p^{\mathrm{SE}}_k, {d}_{k}(p^{\mathrm{SE}}_k,\mathbf{p}_{-k}))
=&\max \limits_{p_k}  W_k(p_k, {d}_{k}(p_k,\mathbf{p}_{-k})). \nonumber
\end{split}
\end{equation}
\end{Definition}

\begin{Definition}\label{nas-1}
(Leader-based NE) $\widetilde{\mathbf{p}}^{\mathrm{NE}}  = ({p}^{\mathrm{NE}}_{2},..., {p}^{\mathrm{NE}}_{N})^{\top}$ is named as an NE of normal UCs if $W_k(p_1,p_2^{\mathrm{NE}},..., p_{k-1}^{\mathrm{NE}}, p_k,p_{k+1}^{\mathrm{NE}},...,p_N^{\mathrm{NE}}) \leq W_k(p_1,\widetilde{\mathbf{p}}^{\mathrm{NE}})$ holds, $\forall k \in \widetilde{\mathcal{N}}$.
\end{Definition}

Definitions \ref{sta2-1} and \ref{nas-1} are modified from Definitions \ref{sta2} and \ref{nas}, respectively, by treating $p_1$ as certain constant. Then, the NE of normal UCs with respect to $p_1$ can be solved by following the structure of (\ref{lp3}), which gives
\begin{equation}\label{441}
\mathbf{A}(p_1,p_2,...,p_N)^{\top} = \mathbf{q}.
\end{equation}
Hence, $\widetilde{\mathbf{p}}^{\mathrm{NE}}$ can be expressed as a function of $p_1$, which is
\begin{align}\label{41}
 & \widetilde{\mathbf{p}}^{\mathrm{NE}}  =  \tilde{\mathbf{f}}(p_1) = \widetilde{\mathbf{A}} (\tilde{\mathbf{q}} - \tilde{\mathbf{s}}p_1),
\end{align}
where
\begin{align}
\widetilde{\mathbf{A}} = \left[
      \begin{array}{ccc}
        s_{22} & \cdots & s_{2N} \\
        \vdots & \ddots & \vdots \\
        s_{N2} & \cdots & s_{NN} \\
      \end{array}
    \right]^{-1}, \tilde{\mathbf{q}} =
\left[
  \begin{array}{c}
    q_2 \\
    \vdots \\
    q_N \\
  \end{array}
\right], \tilde{\mathbf{s}}=
 \left[
   \begin{array}{c}
     s_{21} \\
     \vdots \\
     s_{N1} \\
   \end{array}
 \right], \nonumber
\end{align}
with $\tilde{\mathbf{f}}(p_1) = ({f_2}(p_1),...,{f_N}(p_1))^{\top}$ a functional vector. Essentially, (\ref{41}) is obtained by solving $\widetilde{\mathbf{p}}^{\mathrm{NE}}$ in (\ref{441}) by viewing $p_1$ as known. Therefore, as seen from (\ref{41}), $\widetilde{\mathbf{p}}^{\mathrm{NE}}$ can be driven to some desired positions by designing appropriate $p_1$. By plugging $(p_1,{p}^{\mathrm{NE}}_{2},...,{p}^{\mathrm{NE}}_{N})^{\top}$ into (\ref{xww}) in the sense of (\ref{41}), a revised optimization problem with variable $p_1$ can be formulated as
\begin{align}\label{xg}
{\mathbf{(P4)}}  \quad &\max \quad \widetilde{S}(p_1),
\end{align}
where $\widetilde{S}(p_1)$ is defined as the social profit function with independent variable $p_1$ to distinguish it from $S(\mathbf{p})$.

In the proposed scheme, the price change of $p_1$, defined by $\mu_1$, is decided by ISO from the social profit's point of view. Note that $\widetilde{\mathbf{p}}^{\mathrm{NE}}$ is an affine function of $p_1$ in (\ref{41}). Hence, the interaction between $p^{\mathrm{NE}}_k$ and $p_1$ can be viewed as a leader-following problem and the price change of follower $k$, defined by $\mu_k$, can be expressed by
\begin{align}\label{fi-1}
\widetilde{\bm{\mu}}  = \nabla \tilde{\mathbf{f}}(p_1) \mu_1 = - \widetilde{\mathbf{A}}\tilde{\mathbf{s}} \mu_1,
\end{align}
where $\widetilde{\bm{\mu}} =  (\mu_2,...,       \mu_N)^{\top} $. In addition, we define
\begin{equation}\label{fi}
  (\theta_2,...,\theta_N)^{\top}= - \widetilde{\mathbf{A}}\tilde{\mathbf{s}},
\end{equation}
where $\theta_k$ can reflect the leader-following sensitivity of follower $k$ and provides a measurement of the follow-up step-size. Therefore, in a single timescale from instants $t$ to $t+1$, the evolution of $p_1$ and $\widetilde{\mathbf{p}}^{\mathrm{NE}}$ can be described in {Algorithm 2}.

\begin{algorithm}[H]
\caption{State evolution in a single timescale by Nash-Stackelberg game}\label{a2}
\begin{algorithmic}[1]
\State For $t= 0,1,2,...$:
\State Evolution of {g-UC}: $p^{t}_1 \rightarrow p^{t+1}_1$;
\State Evolution of {normal UCs}:
\begin{align}\label{}
\widetilde{\mathbf{p}}^{NE,t}  \rightarrow \widetilde{\mathbf{p}}^{t} (1) \rightarrow \widetilde{\mathbf{p}}^{t} (2) \rightarrow ... \rightarrow \widetilde{\mathbf{p}}^{NE,t+1}. \nonumber
\end{align}
\end{algorithmic}
\end{algorithm}

\begin{Remark}\label{rr2}
In an incomplete communication graph \cite{salehisadaghiani2018distributed}, $p^{t+1}_1$ can be announced to normal UCs along the communication links first (this task is not difficult since $p^{t+1}_1$ is unchanged before the end of a single timescale). Then with the known $p^{t+1}_1$, the NE of normal UCs $\widetilde{\mathbf{p}}^{NE,t+1}$ can be obtained in a distributed manner, e.g., by the algorithm provided in \cite{ye2017distributed}.
\end{Remark}

\subsubsection{Optimization Strategy} With the updating of $p_1$, a multi-timescale leader-following process can be obtained for $\widetilde{\mathbf{p}}^{\mathrm{NE}}$. Based on {Algorithm 2}, in the following discussion, we focus on the outer-loop strategies by assuming that the inner-loop evolutions can be completed in a single timescale. Thus, to avoid confusions, the outer-loop states of g-UC and normal UCs are redefined by $w_1$ and $w_k$, $k \in \widetilde{\mathcal{N}}$, respectively, i.e., $w_1=p_1$, $w_k = {p}^{\mathrm{NE}}_k$, $\widetilde{\mathbf{w}} = \widetilde{\mathbf{p}}^{\mathrm{NE}}$, and $\mathbf{w} = (w_1,\widetilde{\mathbf{w}}^{\top})^{\top}$. Based on (\ref{fi-1}) and (\ref{fi}), in a multi-timescale horizon, we have
\begin{align}\label{tr}
w_k^{t+1} - w_k^{t} = \theta_k (w_1^{t+1} - w_1^{t}), \quad \forall k \in \widetilde{\mathcal{N}}.
\end{align}

The updating algorithm of g-UC is designed as
\begin{align}\label{12ss}
w^{t+1}_1= w^t_1+ \xi(\mathbf{w}^{t}),
\end{align}
where $\xi(\mathbf{w}^t) = (\nabla S^t)^{\top} \bm{\mu}$, $\bm{\mu} =(\mu_1 ,...,\mu_N)^{\top}$, and $S^t = S(\mathbf{w}^t)$ for simplicity. In addition, one can also have
\begin{align}\label{xxb}
\xi(\mathbf{w}^t)  = (\nabla S^t)^{\top} \bm{\eta}\mu_1,
\end{align}
where (\ref{fi-1}) and (\ref{fi}) are used, $\bm{\eta}=(1,\theta_2,...,\theta_N)^{\top}$. Based on (\ref{tr}) and (\ref{12ss}), the updating law of normal UC $k$ is
\begin{align}\label{12ss+1}
w^{t+1}_k = w^t_k + \theta_k \xi(\mathbf{w}^{t}), & \quad \forall k \in \widetilde{\mathcal{N}}.
\end{align}
Then, a compact form of (\ref{12ss}) and (\ref{12ss+1}) can be written as
\begin{align}\label{xcx+1}
\mathbf{w}^{t+1} = \mathbf{w}^{t} + \bm{\eta} \xi(\mathbf{w}^{t}).
\end{align}

{The proposed multi-timescale leader-following framework allows the followers to continuously update their strategies during a single timescale. When the quasi steady states are achieved, the leader performs the subsequential update. The overall updating process of $\mathbf{w}$ is illustrated in Fig. \ref{1as}, which essentially extends Algorithm 2 to multiple timescales.}

\begin{figure}
  \centering
  \includegraphics[width=9cm]{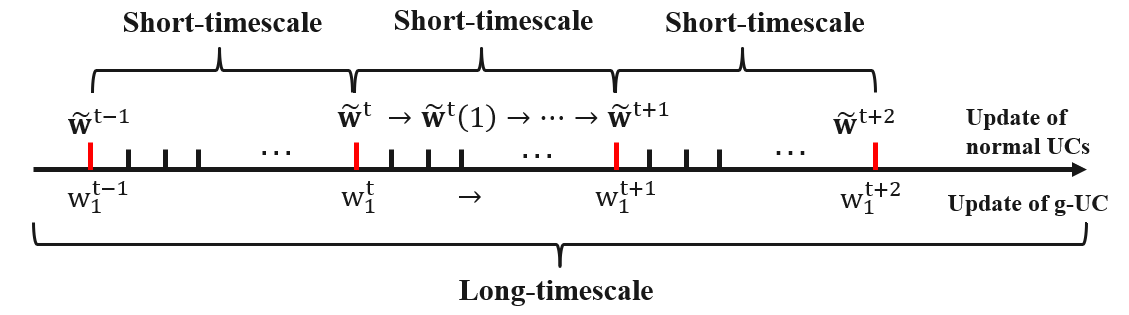}\\
  {\caption{An illustration of the update of $\mathbf{w}$ in multiple timescales. In each short-timescale, $w_1$ is updated first. Then, $\widetilde{\mathbf{w}}$ follows up based on Nash games and the updated $w_1$. When the NE is achieved, a similar updating process of $\mathbf{w}$ is conducted in the next short-timescale.}\label{1as}}
\end{figure}

\begin{Lemma}\label{52}
If $w^*_1$ is a steady state of (\ref{12ss}) with $\mu_1 > 0$, then $w^*_1=\arg\max\limits \widetilde{S}(w_1)$.
\end{Lemma}

\begin{proof}
See {Appendix \ref{521}}.
\end{proof}
Based on {Lemma \ref{52}} and (\ref{tr}), given the optimal strategy $w^*_1$ of the leader, the steady state of followers $\widetilde{\mathbf{w}}^*=(w_2^*,...,w_N^*)^{\top}$ can be characterized by
\begin{align}\label{tr1}
w_k^*- w_k^{t} = \theta_k (w_1^* - w_1^{t}), \quad \forall k \in \widetilde{\mathcal{N}}.
\end{align}

\begin{Theorem}\label{213}
Suppose that $\bm{\eta}^{\top} \mathbf{H}_S \bm{\eta} \neq 0$, where $\mathbf{H}_S$ is the Hessian matrix of $S$. By (\ref{12ss}) and (\ref{12ss+1}), $w_1$ and $\widetilde{\mathbf{w}}$ converge to $w^*_1$ and $\widetilde{\mathbf{w}}^*$ at a linear rate, respectively, if $\mu_1$ is chosen within the range determined by $\mu_1  \bm{\eta}^{\top} \mathbf{H}_S \bm{\eta} \in (-2,0)$.
\end{Theorem}

\begin{proof}
See {Appendix \ref{2131}}.
\end{proof}

\begin{Remark}
In {Theorem \ref{213}}, the validity of $\bm{\eta}^{\top} \mathbf{H}_S \bm{\eta} \neq 0$ is decided by the system parameters, which ensures that the social profit is improvable under the influence of the leader. In the trivial case that $\bm{\eta}^{\top} \mathbf{H}_S \bm{\eta} = 0$, the social profit will remain unchanged with any change of the leader's price. In addition, since $\mathbf{H}_S$ is negative semi-definite (see the proof of {Lemma \ref{ov}}), $\mu_1$ solved by {Theorem \ref{213}} is always positive.
\end{Remark}

The optimization procedure in the basic leader-following framework is stated in {Algorithm 3}.

\begin{algorithm}[H]
\caption{Basic multi-timescale leader-following optimization algorithm}
\begin{algorithmic}[1]
\State Initialize price vector $\mathbf{w}^0$ ($\widetilde{\mathbf{w}}^0$ is an NE with certain $w^0_1$);
\For {$t=0,1,...$}
\For {$k=1,2,...N$} (in parallel)
\State Update $w^{t+1}_1$ and $w^{t+1}_k$ by (\ref{12ss}) and (\ref{12ss+1}), respectively, $k \in \widetilde{\mathcal{N}}$;
\EndFor;
\EndFor;
\end{algorithmic}
\end{algorithm}

\begin{Remark}
(\ref{12ss}) and (\ref{12ss+1}) reflect a decentralized realization, where the ISO only needs to provide $w_1$ to the g-UC and thus can avoid the overload of the coordination center and the associated communication infrastructures \cite{cai2016distributed}. In practice, the multi-timescale leader-following scheme is applicable in some repeated biding markets, such as some markets in Germany \cite{germ}. The optimal solution is achieved after sufficiently many bidding rounds. {When applying Algorithm 3 during multiple optimization intervals, e.g., hourly optimizations during a day \cite{yu2015real}, a single short-timescale indicates an update of (\ref{xcx+1}) and the optimal pricing strategy of UCs in each interval is the steady state of (\ref{xcx+1}).} However, one can note that Algorithm 3 cannot guarantee PoA to be 1 due to the different natures of {Problems (P3)} and {(P4)}.
\end{Remark}

\subsection{DFA Based Multi-timescale Leader-following Framework}\label{sp}

In this section, we propose a DFA strategy, by which the maximum market efficiency can be achieved (PoA=1). In this scheme, after collecting the demand functions of users, the ISO works out the corresponding ``ameliorated'' functions and send them to UCs. The rationality of this behavior is explained as follows.
\begin{enumerate}
  \item A social-centric ISO has the responsibility to improve the social welfare of the market by making compulsory market rules \cite{shoreh2016survey,weidlich2008critical}.
  \item In practice, the proposed coordination strategy can be realized completely imperceptibly or by signing contracts with users and UCs in advance.
\end{enumerate}
In this process, UCs first optimize their own objective functions with the ameliorated demand functions provided by the ISO. Then, the ISO clears the market and distributes the energy product to users with {{real demand functions}}. Note that ISO does not change the selfish instinct of normal UCs since they still compute NE by using the virtual demand functions.

\subsubsection{DFA Based Leader-following Problem} We design $\lambda_{k}$ as the coefficient of $p_1$ (i.e., $w_1$) in the real demand function (\ref{con1}), $k \in \widetilde{\mathcal{N}} $. Then, the ameliorated demand functions of user $i$ are designed as
\begin{align}
  & \bar{d}_{i,k}=  \frac{1}{\beta}(\frac{1}{N}\sum_{j \in \widetilde{\mathcal{N}} }w_j+ \frac{1}{N} \lambda_{k} w_1 -w_k )+\frac{y_{i}}{N}, \label{con3} \\
  & \bar{d}_{i,1}=y_i-\sum_{k \in \widetilde{\mathcal{N}} } \bar{d}_{i,k}. \label{con4}
\end{align}
(\ref{con4}) implies that the stipulated demand $y_{i}$ remains unchanged after the emendations. Similar to (\ref{gge}), the profit function of UC $k$ can be written as
\begin{equation}\label{gge1}
{W}_{k}(p_k, \bar{d}_{k})=p_k \bar{d}_{k}- D_k(\bar{d}_k),
\end{equation}
where $\bar{d}_{k}=\sum_{i \in \mathcal{M}}\bar{d}_{i,k}$. Then, with the newly obtained objective function (\ref{gge1}), a revised leader-based NE among normal UCs can be obtained by the same derivation procedure introduced in Section \ref{eb}. Note that $\lambda_k$ only exists with $w_1$, then the NE can be straightforwardly obtained by following the structure of (\ref{41}), which gives
\begin{align}\label{dee1}
    \widetilde{\mathbf{w}}= \widetilde{\mathbf{A}}(\tilde{\mathbf{q}}-\mathrm{diag}[ {\tilde{\mathbf{s}}}]\widetilde{\bm{\lambda}} w_1),
\end{align}
where $\widetilde{\bm{\lambda}}=(\lambda_2,...,\lambda_{N})^{\top}$ and $\mathrm{diag}[ {\tilde{\mathbf{s}}}]$ is a diagonal matrix with the elements of ${\tilde{\mathbf{s}}}$ placed on the diagonal.

\subsubsection{DFA Based Optimization Strategy} In the following, we will give suitable $\widetilde{\bm{\lambda}}^{t+1}$ for bidding round $t+1$ for normal UCs such that the market efficiency is maximized. With (\ref{dee1}), the updating law of normal UCs in bidding round $t+1$ can be written as
\begin{align}\label{dee2}
    \widetilde{\mathbf{w}}^{t+1}= \widetilde{\mathbf{A}}(\tilde{\mathbf{q}}-\mathrm{diag}[ {\tilde{\mathbf{s}}}]\widetilde{\bm{\lambda}}^{t+1} w^{t+1}_1),
\end{align}
and the updating law for $\widetilde{\bm{\lambda}}^{t+1}$ is designed as
\begin{align}\label{ed}
& \widetilde{\bm{\lambda}}^{t+1} = \Phi^{t+1} \widetilde{\mathbf{A}}^{-1} \Psi^{t+1}, \\
& \Phi^{t+1}=
\left[
  \begin{array}{ccc}
    \dfrac{1}{{s}_{21}w_1^{t+1}} &   & \mathbf{0} \\
      & \ddots &   \\
    \mathbf{0} &   & \dfrac{1}{{s}_{N1}w_1^{t+1}} \\
  \end{array}
\right], \\
&
\Psi^{t+1} =
 \left[
   \begin{array}{c}
     \sum_{v=0}^t \dfrac{\nabla_2 S^v}{\nabla_1 S^v}(w^v_1 - w^{v+1}_1) \\
     \vdots \\
     \sum_{v=0}^t \dfrac{\nabla_N S^v}{\nabla_1 S^v}( w^v_1 -w^{v+1}_1   ) \\
   \end{array}
 \right].
\end{align}
(\ref{dee2}) implies that in round $t+1$, $\widetilde{\bm{\lambda}}^{t+1}$ and $w^{t+1}_1$ will be updated and announced ahead of the actions of normal UCs.
Then the follow-up step-size of normal UCs can be calculated by
\begin{align}\label{ed2}
 \widetilde{\mathbf{w}}^{t+1} -  \widetilde{\mathbf{w}}^{t} & = \widetilde{\mathbf{A}} \mathrm{diag}[ {\tilde{\mathbf{s}}}] \widetilde{\bm{\lambda}}^{t} w^t_1 - \widetilde{\mathbf{A}} \mathrm{diag}[ {\tilde{\mathbf{s}}}]\widetilde{\bm{\lambda}}^{t+1} w^{t+1}_1 \nonumber \\
& =
\left[
  \begin{array}{c}
     \dfrac{\nabla_2 S^t}{\nabla_1 S^t}(w^{t+1}_1-w^t_1)\\
    \vdots \\
    \dfrac{\nabla_N S^t}{\nabla_1 S^t}(w^{t+1}_1-w^t_1) \\
  \end{array}
\right].
\end{align}
(\ref{ed2}) defines a leader-following relationship with varying sensitivity vector $\bar{\bm{\theta}}^t=(\bar{\theta}^t_2,...,\bar{\theta}^t_N)^{\top} \triangleq ( \frac{\nabla_2 S^t}{\nabla_1 S^t},...,\frac{\nabla_N S^t}{\nabla_1 S^t}  )^{\top}$.

With the DFA strategy, the updating algorithm of g-UC is designed as
\begin{align}\label{12ss+2}
w^{t+1}_1= w^t_1+ \bar{\xi}(\mathbf{w}^{t}),
\end{align}
where $\bar{\xi}(\mathbf{w}^t) = (\nabla S^t)^{\top} \bar{\bm{\eta}}^t \mu^t_1$ with $\bar{\bm{\eta}}^t = (1,\bar{\theta}^t_2,..., \bar{\theta}^t_N)^{\top}$.
Based on (\ref{ed2}) and (\ref{12ss+2}), the updating law for normal UCs is
\begin{align}\label{12ss+3}
w^{t+1}_k = w^t_k + \bar{\theta}^t_k \bar{\xi}(\mathbf{w}^{t}),\quad \forall k \in \widetilde{\mathcal{N}}.
\end{align}
Then, a compact form of (\ref{12ss+2}) and (\ref{12ss+3}) can be written as
\begin{align}\label{xcx+2}
\mathbf{w}^{t+1} = \mathbf{w}^{t} + \bm{\bar{\eta}}^t \bar{\xi}(\mathbf{w}^{t}).
\end{align}

\begin{Lemma}\label{p1}
If ${\bar{\xi}}(\mathbf{w}^t) \rightarrow 0$ and $\mu_1 > 0$, we have $\nabla_l S(\mathbf{w}^t) \rightarrow 0$, $\forall l \in \mathcal{N}$, $t=0,1,2,...$
\end{Lemma}

\begin{proof}
See {Appendix \ref{pp1}}.
\end{proof}

\begin{Theorem}\label{2133}
By (\ref{xcx+2}), $w_1$ and $\widetilde{\mathbf{w}}$ converge to the optimal solution to {Problem (P3)} if $\mu^t_1 \parallel \bar{\bm{\eta}}^t \parallel^2 = \frac{1}{L}$,
where $L$ is the Lipschitz constant of $\nabla S$, $t=0,1,2,..$.
\end{Theorem}

\begin{proof}
See {Appendix \ref{21331}}.
\end{proof}

The DFA based optimization procedure is stated in {Algorithm 4}.
\begin{algorithm}[H]
\caption{DFA based multi-timescale leader-following optimization algorithm}
\begin{algorithmic}[1]
\State Initialize price vector $\mathbf{w}^0$ ($\widetilde{\mathbf{w}}^0$ is an NE with certain $w^0_1$);
\For {$t=0,1,...$}
\For {$k=1,2,...N$} (in parallel)
\State Update $w^{t+1}_1$ and $w^{t+1}_k$ by (\ref{12ss+2}) and (\ref{12ss+3}), respectively, $k \in \widetilde{\mathcal{N}}$;
\EndFor;
\EndFor;
\end{algorithmic}
\end{algorithm}

{\begin{Remark}\label{re}
{The distributed realization of Algorithm 4 can be similar to Algorithm 3 as discussed in Remark \ref{rr2}, except that the ``ameliorated'' demand functions are needed to be announced to UCs efficiently in each short-timescale since the coefficient $\widetilde{\bm{\lambda}}$ varies with $t$ in (\ref{ed}).} Therefore, we can set the ISO as the central node, who is linked to all the UCs. In addition, by the proposed DFA strategy, the true demand functions are only kept by ISO and unknown to UCs, which overcomes the privacy-releasing issue of conventional Stackelberg game approaches.
\end{Remark}}

\subsection{Optimization Strategy with Constraints}\label{c}

In the following, we consider convex constraints $p_1 \in \mathcal{P}_1$ and $\widetilde{\mathbf{p}}= (p_2,...,p_N)^{\top}  \in \widetilde{\mathcal{P}}$ for g-UC and normal UCs, respectively, with $\mathcal{P} = \mathcal{P}_1\times \widetilde{\mathcal{P}}$. Then, Problem (P2) is modified into a constrained NE problem
\begin{align}\label{gene1}
{\mathbf{(P5)}}  \quad \max\limits_{p_k}  \quad & W_k({p_k},\mathbf{p}_{-k}), \quad \forall k \in \mathcal{N} \nonumber \\
\hbox{subject to} \quad & \mathbf{p} \in \mathcal{P}.
\end{align}
Given certain $p_1$, the optimization problem of normal UCs can be formulated as
\begin{align}\label{gene2}
{\mathbf{(P6)}}  \quad \max\limits_{p_k}  \quad & W_k({p_k},\mathbf{p}_{-k}), \quad \forall k \in \widetilde{\mathcal{\mathcal{N}}} \nonumber \\
\hbox{subject to} \quad & \widetilde{\mathbf{p}} \in \widetilde{\mathcal{P}}.
\end{align}
In {Problem (P6)}, the NE of normal UCs can be nonlinearly determined by $p_1$. Hence, the newly obtained social profit optimization problem of the leader can be formulated as
\begin{align}
{\mathbf{(P7)}}  \quad  \max\limits_{p_1 \in \mathcal{P}_1} \quad  & S (p_1,\widetilde{\mathbf{p}}^{\mathrm{NE}}) \nonumber \\
\hbox{subject to} \quad & \widetilde{\mathbf{p}}^{\mathrm{NE}} \in  \mathrm{Arg} \hbox{ Problem (P6)}. \label{44+1}
\end{align}
It can be noted that {Problem (P7)} is a bi-level optimization problem and is possibly non-convex in $p_1$, whose global optimal solution can be difficult to obtain mathematically, as discussed by many works on similar bi-level optimization problems \cite{cui2018two,ma2016energy,savard1994steepest}. In this work, we use the affine function $\tilde{\mathbf{f}}$ defined in (\ref{41}) to approximate the nonlinear relation (\ref{44+1}). Then, the constraint on $p_1$ will be addressed by projecting the result in each updating round onto $\mathcal{P}_1$, which is technically easy to solve and produces an approximate optimal result. Then, based on (\ref{12ss}), the updating algorithm of g-UC with constraint $\mathcal{P}_1$ is designed as
\begin{align}\label{33+1}
w^{t+1}_1 & = \Omega_{\mathcal{P}_1} [w^t_1+ \xi(\mathbf{w}^{t})] \nonumber \\
& = \Omega_{\mathcal{P}_1} [w^t_1+ \mu_1 \nabla \widetilde{S}(w^t_1)],
\end{align}
where (\ref{33}) in the Appendix is used and $\Omega_{\mathcal{P}_1}[\cdot]$ is an Euclidean projection onto $\mathcal{P}_1$. Then based on (\ref{xcx+1}), a compact form of the updating algorithm of UCs can be written as
\begin{align}\label{w1w}
\mathbf{w}^{t+1} = & \Omega_{\mathcal{P}} [\mathbf{w}^t + \bm{\eta} \xi(\mathbf{w}^{t})],
\end{align}
whose convergence is closely related to the dynamics of $w_1$ in (\ref{33+1}). The convergence of (\ref{33+1}) can be ensured by letting $\mu_1= \frac{1}{L_1}$,
where $L_1$ is the Lipschitz constant of $\nabla \widetilde{S}(w_1)$ \cite[Thm. 3.7]{bubeck2015convex}.
Similar modifications on (\ref{xcx+2}) can be made with the help of (\ref{yy}) in the Appendix. A projected updating algorithm of (\ref{xcx+2}) is given as
\begin{align}\label{xcx+3}
\mathbf{w}^{t+1} = & \Omega_{\mathcal{P}} [\mathbf{w}^{t} + \bm{\bar{\eta}}^t \bar{\xi}(\mathbf{w}^{t})] \nonumber \\
= & \Omega_{\mathcal{P}} [\mathbf{w}^{t} +  \mu_1^t \parallel \bar{\bm{\eta}}^t \parallel^2 \nabla S^t].
\end{align}
The convergence of (\ref{xcx+3}) can be ensured by choosing the step-size $\mu_1^t$ determined by $\mu_1^t \parallel \bar{\bm{\eta}}^t \parallel^2= \frac{1}{L}$, where $L$ is the Lipschitz constant of $\nabla S$ \cite[Thm. 3.7]{bubeck2015convex}.

In the following, we will discuss how to formulate some physical constraints for formula (\ref{xcx+3}) by considering ramp rate limits, generation limits, and capacity limits of power lines.

\subsubsection{Ramp Rate Limits}

The output of UC $k$ may be affected by the output of the previous instant due to the ramp rate limits \cite{niknam2012new}. To address this issue, we define the up and down ramp rate limits of UC $k$ by $\overline{\delta}_k>0$ and $\underline{\delta}_k<0$, respectively. In addition, we let the output of UC $k$ in the previous time interval be $d'_{k}$, which is known in the current interval. Then, the ramp rate constraints can be formulated as
\begin{align}
& d_k - d'_{k} \geq \Delta \underline{\delta}_k , \label{3.1} \\
& d_k - d'_{k} \leq \Delta \overline{\delta}_k , \label{3.2}
\end{align}
where $\Delta$ is the length of certain operation interval. Based on (\ref{con2}), at the NE, the energy supply of UC $k$ can be given by
\begin{align}
d_{k} = \frac{M}{\beta}(\frac{1}{N}\sum_{j \in \mathcal{N}} w_j-w_k )+\frac{Y}{N}. \label{3.5}
\end{align}
Therefore, the feasible region under the ramp rate limits of UC $k$ can be obtained by combining (\ref{3.1})-(\ref{3.5}), which gives
\begin{align}\label{3.7}
\mathcal{P}^{\mathrm{r}}_k= \{ \mathbf{w} | \frac{M}{\beta}(\frac{1}{N} \sum_{j \in \mathcal{N}}w_j  -w_k )+\frac{Y}{N} \nonumber \\
 \in [d'_{k} + \Delta \underline{\delta}_k,d'_{k} + \Delta \overline{\delta}_k]\}.
\end{align}
Therefore, the overall ramp rate constraint can be obtained as $\mathcal{P}^{\mathrm{r}}= \bigcap_{k \in \mathcal{N}} \mathcal{P}^{\mathrm{r}}_k$.

\subsubsection{Generation Limits}

The upper and lower bounds of the output of UC $k$ are defined by $u_k$ and $l_k$, respectively, which means the feasible region of $d_k$ is
\begin{align}\label{cl}
d_k \in [l_k,u_k].
\end{align}
{Then based on (\ref{3.5}) and (\ref{cl}), the feasible region of the generation limits of UC $k$ can be obtained as}
\begin{align}\label{551}
\mathcal{P}^{\mathrm{g}}_k= \{ \mathbf{w} | \frac{M}{\beta}(\frac{1}{N} \sum_{j \in \mathcal{N}}w_j -w_k ) +\frac{Y}{N} \in [l_k,u_k] \}.
\end{align}
Then, the feasible region under the overall generation limits is $\mathcal{P}^{\mathrm{g}}= \bigcap_{k \in \mathcal{N}} \mathcal{P}^{\mathrm{g}}_k$.

\subsubsection{Capacity Limits of Power Lines}

The capacity constraint of the power line for carrying demand $d^*_{i,k}$ is defined by
\begin{align}\label{co0}
d^*_{i,k} \in [\underline{\rho}_{i,k}, \overline{\rho}_{i,k}].
\end{align}
At the NE, based on (\ref{con1}), we have
\begin{equation}\label{co1}
d^*_{i,k}=\frac{1}{\beta}(\frac{1}{N}\sum_{j \in \mathcal{N}}w_j -w_k )+\frac{y_{i}}{N}.
\end{equation}
Then, by (\ref{co0}) and (\ref{co1}), the capacity limits of the power line between user $i$ and UC $k$ can be given by
\begin{align}\label{3.9}
\mathcal{P}^{\mathrm{c}}_{i,k}= \{ \mathbf{w} | \frac{1}{\beta}(\frac{1}{N}  \sum_{j \in \mathcal{N}} w_j  -w_k )+\frac{y_i}{N} \in [\underline{\rho}_{i,k}, \overline{\rho}_{i,k}] \}.
\end{align}
Therefore, the feasible region under the overall capacity limits can be obtained as $\mathcal{P}^{\mathrm{c}}= \bigcap_{i \in \mathcal{M}} \bigcap_{k \in \mathcal{N}} \mathcal{P}^{\mathrm{c}}_{i,k}$.
With this settlement, the rationality condition of $\mathbf{d}^*_i$ discussed in Corollary \ref{nb} can be satisfied by letting $\underline{\rho}_{i,k} = 0$, $\forall k \in \mathcal{N}$.}

In some scenarios, to realize the social optimality, it is possible to have negative optimal profits for some UCs in the sense that the social optimal solution that benefits all the agents may not always exist depending on the specific parameters of the market. For instance, as a result of competition, some UCs may suffer budget deficits under the influence of market manager if their cost profiles are not good enough. To address this issue, there are two possible solutions: {{(i)}} since the social profit is improved, the ISO may redistribute the profit from the beneficiaries to the agents with losses without affecting the social net profit of the market; {{(ii)}} one may set limitations on energy prices and generation limits for UCs. {For example, based on (\ref{cl}) and (\ref{551}), one can consider $\mathcal{P}^{\mathrm{p}}_k = \{w_k | w_k \geq \sup_{d_k \in [l_k,u_k]} ( a_k d_k +b_k+\frac{c_k}{d_k} ) \}$, such that $W_k (p_k,d_k)$ in (\ref{gge}) is non-negative with $p_k \in \mathcal{P}^{\mathrm{p}}_k$ and $d_k \in [l_k,u_k]$, $l_k,u_k > 0$. Then the overall feasible region is $\mathcal{P}^{\mathrm{gp}} = \mathcal{P}^{\mathrm{g}} \bigcap \mathcal{P}^{\mathrm{p}}$, where $ \mathcal{P}^{\mathrm{p}} = \mathcal{P}^{\mathrm{p}}_1 \times ... \times \mathcal{P}^{\mathrm{p}}_N$.}

\section{Numerical results}

{In this section, the performance of the proposed algorithms is verified by an IEEE 9-bus system model.}
\begin{figure}[H]
  \centering
  \includegraphics[width=5cm]{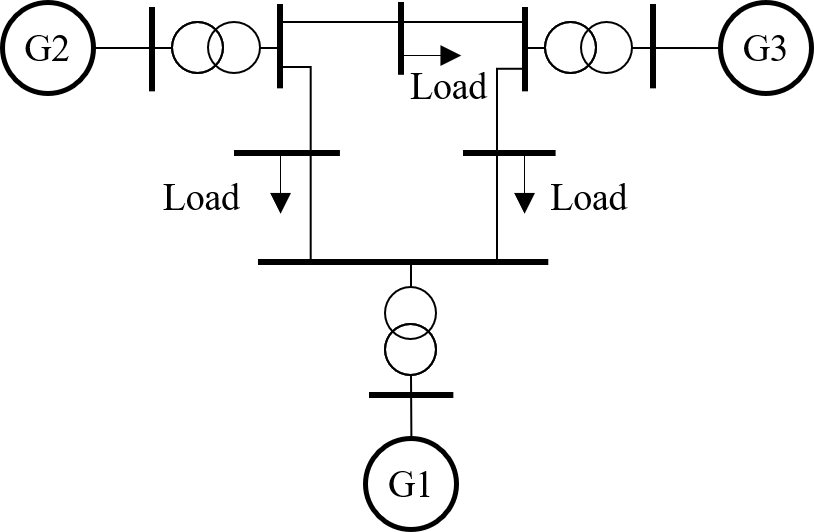}\\
  {\caption{A simplified illustration of an IEEE 9-bus system.}\label{9bus}}
\end{figure}

{\subsection{Simulation Setup}

A simplified illustration of an IEEE 9-bus system is shown in Fig. \ref{9bus}. In this system, there are 3 generators, i.e., G1 to G3, which are assumed to be maintained by UC 1 to UC 3, respectively. In addition, the overall load is assumed to be shared by 5 different energy users. To realize the leader-following mechanism, we set UC 1 as the g-UC with UCs 2 and 3 being the normal UCs.

In this simulation, we will apply our algorithms during 24 hours in certain day (i.e., 24 optimization intervals). The parameters of the UCs are shown in Table \ref{t2} \cite{wang2019social}. To set up the ramp rate constraints, the output of generators during the last hour before the day is uniformly set as $4$ and $\Delta$ is set as 1. The parameters of users' utility functions are set as $\alpha=30$ and $\beta=5$ \cite{wang2019social}. The hourly energy demands of users are arbitrarily selected within $[0,6]$. The lower bound of energy demands is uniformly set as $\underline{\rho}_{i,k} = 0$, $\forall i \in \mathcal{M}$, $\forall k \in \mathcal{N}$.}

\begin{table}
{\caption{Parameters of UCs}\label{t2}
\label{tab2}
\begin{center}
\begin{tabular}{cccc}
\bottomrule
\hspace{8mm} & \hspace{4mm} UC 1 \hspace{4mm} & \hspace{4mm} UC 2 \hspace{4mm} & \hspace{4mm} UC 3 \hspace{4mm}  \\
\bottomrule
$a$ & 0.1 & 0.2 & 0.05 \\
\hline
$b$ & 0.2 & 0.5 & 0.1 \\
\hline
$c$ & 0 & 0.1 & 0.2 \\
\hline
$u$ & 2 & 3 & 4 \\
\hline
$l$ & -2 & -2 & -1.5 \\
\hline
$\overline{\delta}$ & 7 & 8 & 9 \\
\hline
$\underline{\delta}$ & 0.5 & 0.5 & 0.5 \\
\bottomrule
\end{tabular}
\end{center}}
\end{table}

{\subsection{Simulation Result}

First, we optimize the social profit of the market by considering the ramp rate limits, generation limits, and lower bound of users' demands. The obtained results and the explanations are summarized as follows.

\begin{enumerate}
  \item The optimal pricing strategies in the day are shown in Fig. \ref{pri}. The outputs of generators in different hours are depicted in Fig. \ref{op}. It can be seen that, due to the high energy demand in the 9th hour, the upper generation limit of G1 is activated.
  \item The ramp rates of the generators in different hours are shown in Fig. \ref{rr}. Notably, in the 18th hour, the up ramp rate of G1 is bounded by the upper limit. In the 15th and 24th hours, the down ramp rate of G3 is bounded by the lower limit.
  \item The energy-purchasing strategies of users are depicted in Figs. \ref{sm21}-(a) to (e). It can be seen that, for all users, the demands from the UCs are negatively correlated with the settled prices, which verifies the rationality of (\ref{con1}).
\end{enumerate}}

\begin{figure}
  \centering
  \includegraphics[width=9cm]{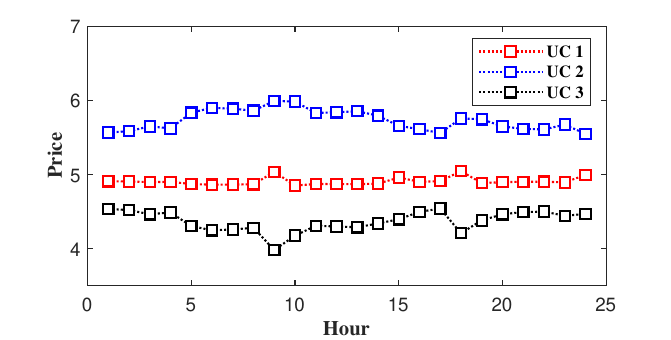}\\
  {\caption{Optimal energy prices in different hours.}\label{pri}}
\end{figure}

\begin{figure}
  \centering
  \includegraphics[width=9cm]{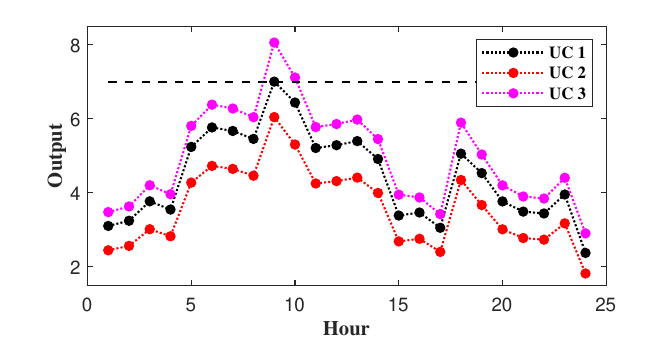}\\
  {\caption{Generation quantities in different hours. The black dash line is the upper generation limit of G1.}\label{op}}
\end{figure}

\begin{figure}
  \centering
  \includegraphics[width=9cm]{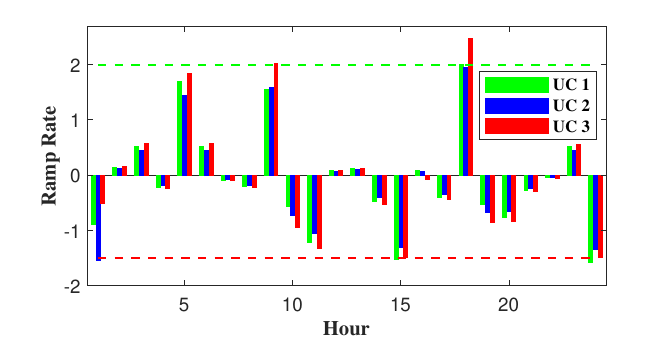}\\
  {\caption{Ramp rates in different hours. A positive value means an up ramp rate, and vice versa. The green and red dash lines are the up and down ramp rate limits of G1 and G3, respectively.}\label{rr}}
\end{figure}

\begin{figure}[hbpt]
\centering
\subfigure[]{
\begin{minipage}[t]{0.47\linewidth}
\centering
\includegraphics[width=4.5cm]{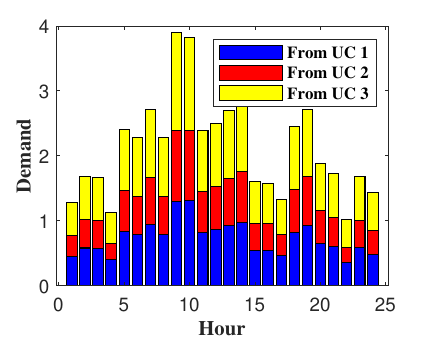}
%\caption{fig1}
\end{minipage}%
}%
\subfigure[]{
\begin{minipage}[t]{0.47\linewidth}
\centering
\includegraphics[width=4.5cm]{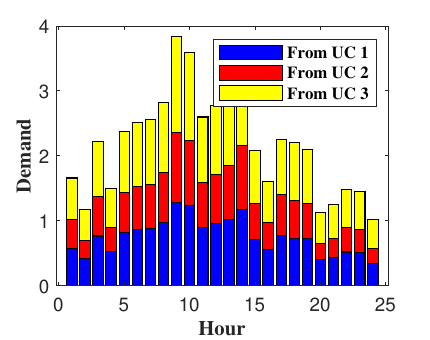}
%\caption{fig2}
\end{minipage}%
} \\ %
\subfigure[]{
\begin{minipage}[t]{0.47\linewidth}
\centering
\includegraphics[width=4.5cm]{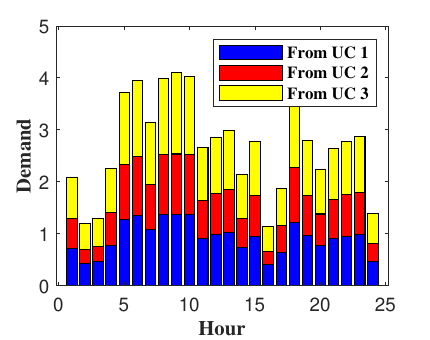}
%\caption{fig2}
\end{minipage}%
}%
\subfigure[]{
\begin{minipage}[t]{0.47\linewidth}
\centering
\includegraphics[width=4.5cm]{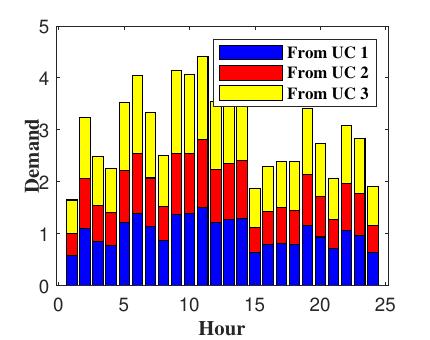}
%\caption{fig2}
\end{minipage}%
}\\ %
\subfigure[]{
\begin{minipage}[t]{0.47\linewidth}
\centering
\includegraphics[width=4.5cm]{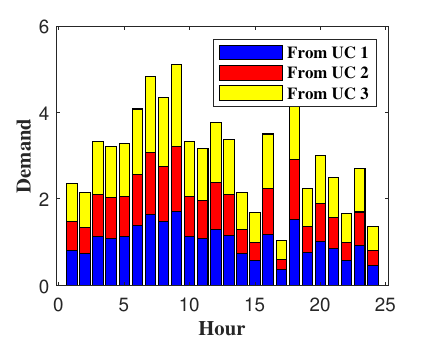}
%\caption{fig2}
\end{minipage}%
}%
{\caption{Optimal energy-purchasing strategies in different hours: (a) User 1; (b) User 2; (c) User 3; (d) User 4; (e) User 5.}\label{sm21}}
\end{figure}

{Second, we demonstrate and analyze the performance of Algorithms 3 and 4 by using the market information in the 7th hour (the physical constraints are inactivated in this hour). The obtained results and the explanations are summarized as follows.
\begin{enumerate}
  \item The simulation result with Algorithm 3 is depicted in Figs. \ref{sm3}-(a) to (d). As shown in Fig. \ref{sm3}-(a), the prices of UCs converge to steady states asymptotically. The price of UC 1 decreases most compared with that of normal UCs. However, the profit of UC 1 is stabilized at a relatively higher level compared with the other UCs. This is because the users would like to purchase more energy from UC 1, whose price is much lower. As seen from Figs. \ref{sm3}-(b) and (c), although the profit of UCs decreases, the profit of users increases significantly. Eventually, as shown in Fig. \ref{sm3}-(d), the value of PoA decreases from around 1.55 to around 1.20, which implies that the social profit of the market is improved.
  \item With the same parameter setting, the simulation result with {Algorithm 4} is shown in Figs. \ref{sm2}-(a) to (d). Fig. \ref{sm2}-(a) shows that the prices tend to steady states asymptotically. Figs. \ref{sm2}-(b) and \ref{sm2}-(c) show that the social profit of the market is improved compared with that obtained by Algorithm 3. As a consequence, the PoA is decreased to 1 as shown in Fig. \ref{sm2}-(d).
\end{enumerate}}

\begin{figure}[hpbt]
\centering
\subfigure[]{
\begin{minipage}[t]{0.47\linewidth}
\centering
\includegraphics[width=4.5cm]{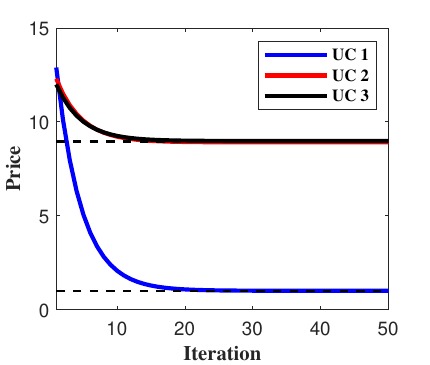}
%\caption{fig1}
\end{minipage}%
}%
\subfigure[]{
\begin{minipage}[t]{0.47\linewidth}
\centering
\includegraphics[width=4.5cm]{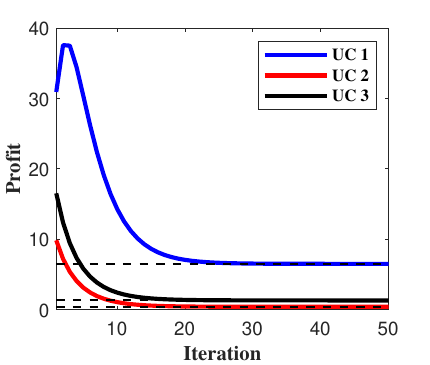}
%\caption{fig2}
\end{minipage}%
} \\ %
\subfigure[]{
\begin{minipage}[t]{0.47\linewidth}
\centering
\includegraphics[width=4.5cm]{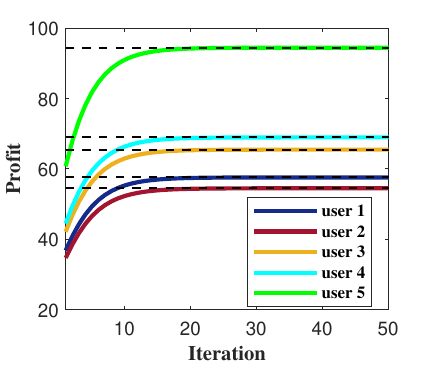}
%\caption{fig2}
\end{minipage}%
}%
\subfigure[]{
\begin{minipage}[t]{0.47\linewidth}
\centering
\includegraphics[width=4.5cm]{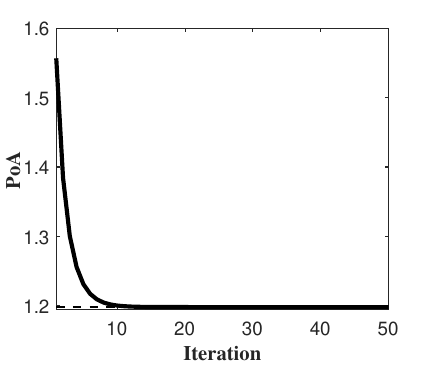}
%\caption{fig2}
\end{minipage}%
}%
{\caption{Optimization result in the 7th hour with Algorithm 3: (a) dynamics of the prices of UCs; (b) dynamics of the profits of UCs; (c) dynamics of the profits of users; (d) dynamics of PoA.}\label{sm3}}
\end{figure}

\begin{figure}[hpbt]
\centering
\subfigure[]{
\begin{minipage}[t]{0.47\linewidth}
\centering
\includegraphics[width=4.5cm]{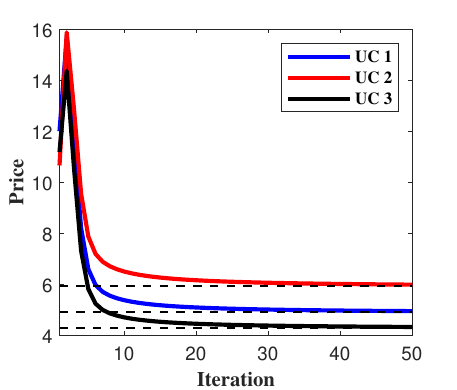}
%\caption{fig1}
\end{minipage}%
}%
\subfigure[]{
\begin{minipage}[t]{0.47\linewidth}
\centering
\includegraphics[width=4.5cm]{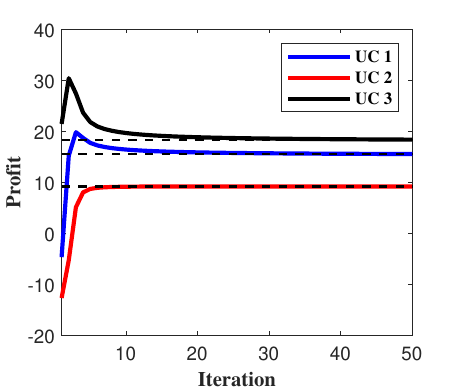}
%\caption{fig2}
\end{minipage}%
}  \\%
\subfigure[]{
\begin{minipage}[t]{0.47\linewidth}
\centering
\includegraphics[width=4.5cm]{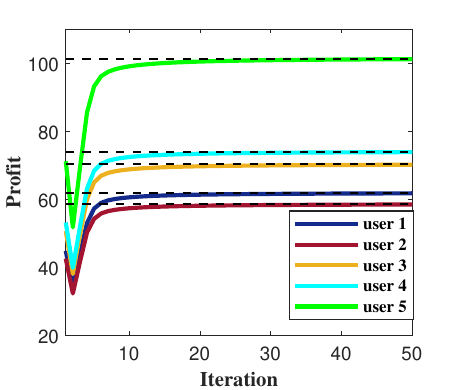}
%\caption{fig2}
\end{minipage}%
} %
\subfigure[]{
\begin{minipage}[t]{0.47\linewidth}
\centering
\includegraphics[width=4.5cm]{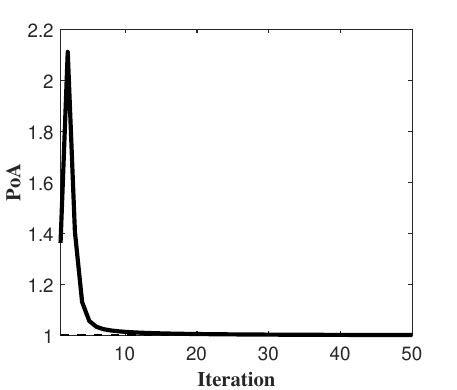}
%\caption{fig2}
\end{minipage}%
}%
{\caption{Optimization result in the 7th hour with Algorithm 4: (a) dynamics of the prices of UCs; (b) dynamics of the profits of UCs; (c) dynamics of the profits of users; (d) dynamics of PoA.}\label{sm2}}
\end{figure}

\section{Conclusions}

In this work, we proposed a multi-timescale leader-following approach for optimizing a multi-UC-multi-user electricity market based on NE and SE analysis. In this model, users aim to optimize their profits by purchasing energy from different UCs. At the UC side, two multi-timescale leader-following algorithms were proposed to optimize the social profit. By considering some additional constraints, two projection based updating algorithms were studied, which can provide approximate optimal solutions for the resulting possibly non-convex optimization problems.
\appendix{}

\subsection{Proof of Theorem 1}\label{conx1}

The dimension of variables of $U_{i}(\mathbf{d}_i)$ can be decreased by selecting a reference $d_{i,\hat{k}}$, which is the demand from a reference UC $\hat{k}$. Then, $d_{i,\hat{k}}= y_i-\sum_{k \in \mathcal{N} \setminus \{\hat{k}\}} d_{i,k}$. By canceling $d_{i,\hat{k}}$, the profit function (\ref{21s}) can be modified into
\begin{align}\label{s1}
  \widetilde{U}_{i} & (\mathbf{d}_{i,-\hat{k}})  = \sum_{l\in \mathcal{N}\setminus \{\hat{k}\}} (\alpha d_{i,l}-\frac{\beta}{2}d^2_{i,l}-p_k d_{i,l} )+\alpha d_{i,\hat{k}} \nonumber \\
& -\frac{\beta}{2}d^2_{i,\hat{k}}-p_{\hat{k}} d_{i,\hat{k}}  \nonumber \\
  = &-\beta d^2_{i,k}+(p_{\hat{k}}-p_k-\beta \sum_{l\in \mathcal{N}\setminus \{k,\hat{k}\}}d_{i,l}+\beta y_i )d_{i,k}\nonumber \\
&+\sum_{l\in \mathcal{N}\setminus \{k,\hat{k}\}} (\alpha d_{i,l}-\frac{\beta}{2}d^2_{i,l}-p_k d_{i,l} ) \nonumber \\
& +(\alpha-p_{\hat{k}})(y_i-\sum_{l\in \mathcal{N}\setminus \{k,\hat{k}\}}d_{i,l} ) \nonumber \\
& -\frac{\beta}{2} (y_i-\sum_{l\in \mathcal{N}\setminus \{k,\hat{k}\}}d_{i,l} )^2,
\end{align}
where $\mathbf{d}_{i,-\hat{k}}=(d_{i,1},...,d_{i,\hat{k}-1}, d_{i,\hat{k}+1},...,d_{i,N})^{\top}$ and $k \in \mathcal{N} \setminus \{\hat{k}\}$. Note that ${U}_i(\mathbf{d}_{i})$ is concave and twice continuously differentiable. Hence, the maximum of $\widetilde{U}_i(\mathbf{d}_{i,-\hat{k}})$ exists and is identical to that of {Problem (P1)} \cite[Sec. 10.1.2]{boyd2004convex}, which can be found by solving the first-order optimality condition
\begin{equation}\label{uir}
    \nabla \widetilde{U}_i(\mathbf{d}_{i,-\hat{k}}) =\mathbf{0},
    \end{equation}
which is in the form of
\begin{equation}\label{e1}
\mathbf{P}\mathbf{d}_{i,-\hat{k}}=\mathbf{b}_{i,-\hat{k}},
\end{equation}
where
\begin{subequations}
\begin{align}\label{}
&\mathbf{P} = \beta(\mathbf{I}_{N-1}+\mathbf{1}_{N-1}\mathbf{1}_{N-1}^{\top}), \\
%   &\mathbf{w}_{i,-\hat{k}}=\left[\begin{array}{ccccc}
%                       2\beta & \beta & \cdots  & \beta \\
%                       \beta & \ddots  & \ddots &   \vdots \\
%                       \vdots & \ddots  & \ddots  &   \beta \\
%                       \beta & \cdots & \beta & 2\beta
%                     \end{array}
%    \right] \in \mathbb{R}^{(N-1)\times(N-1)},\nonumber\\
&\mathbf{b}_{i,-\hat{k}}= (p_{\hat{k}}-p_{1}+\beta y_i,...,p_{\hat{k}}-p_{\hat{k}-1}+\beta y_i,\nonumber\\
& \quad \quad \quad  p_{\hat{k}} -p_{\hat{k}+1} +\beta y_i,...,p_{\hat{k}}-p_{N}+\beta y_i)^{\top}.
\end{align}
\end{subequations}
Since $\mathbf{P}$ is invertible, the solution to (\ref{e1}) can be solved by $\mathbf{d}^*_{i,-\hat{k}}=\mathbf{P}^{-1}\mathbf{b}_{i,-\hat{k}}$, which gives
\begin{equation}\label{con99}
d^*_{i,k}=\frac{1}{\beta} (\frac{1}{N}\sum_{j \in \mathcal{N}}p_j-p_k )+\frac{y_{i}}{N},\quad \forall k \in \mathcal{N} \setminus \{\hat{k}\}.
\end{equation}
Then, the purchased energy from UC $\hat{k}$ is $d^*_{i,\hat{k}} = y_i-\sum_{k \in \mathcal{N}\setminus \{ \hat{k}\}} d^*_{i,k}$. It can be checked that the expression of $d^*_{i,\hat{k}}$ also follows the structure of (\ref{con99}) except substituting subscript $k$ by $\hat{k}$. Hence, the solution to {Problem (P1)} can be written by a uniform formula
\begin{equation}\label{con10}
d^*_{i,k}=\frac{1}{\beta} (\frac{1}{N}\sum_{j \in \mathcal{N}}p_j-p_k )+\frac{y_{i}}{N}, \quad \forall k \in \mathcal{N}.
\end{equation}

\subsection{Proof of Corollary 2}\label{c3p}

Rewriting (\ref{con1}) with Nash price $\mathbf{p}^{\mathrm{NE}}$ gives
\begin{align}\label{40}
\mathbf{d}^*_i = \frac{1}{\beta N} \mathbf{B} \mathbf{p}^{\mathrm{NE}} + \frac{\mathbf{y}_i}{N}.
\end{align}
Then by (\ref{16}), the rationality condition $\mathbf{d}^*_{i} \geq \mathbf{0}$ can be obtained as (\ref{14}).

\subsection{Proof of Lemma \ref{ov}}\label{ov1}
The Hessian matrix of $S(\mathbf{p})$ can be written as
\begin{equation}\label{}
    \mathbf{H}_S=[\nabla_{kk'} S]_{k,k' \in \mathcal{N}}.
\end{equation}
By (\ref{u1}) and (\ref{e2}), $\forall j,k \in \mathcal{N}, j  \neq k$, we can have
\begin{align}\label{}
 \nabla_{kk} \sum_{l \in \mathcal{M}}U^*_l = & \frac{M(N-1)}{\beta N},  \\
\nabla_{kj} \sum_{l \in \mathcal{M}}U^*_l = & -\frac{M}{\beta N},  \\
\nabla_{kk} \sum_{{l'} \in \mathcal{N}}W_{l'}= & -\frac{2a_kM^2(N-1)^2}{N^2\beta^2}-\frac{2M(N-1)}{N\beta} \nonumber \\
& -\sum_{h \in \mathcal{N}\setminus \{k\}}\frac{2a_hM^2}{N^2\beta^2}, \\
\nabla_{kj} \sum_{{l'} \in \mathcal{N}}W_{l'} = & \frac{2(a_k+a_j)M^2(N-1)}{N^2\beta^2}  + \frac{2M}{N\beta}  \nonumber \\
& - \sum_{h \in \mathcal{N}\setminus \{j,k\}}\frac{2a_hM^2}{N^2\beta^2}.
\end{align}
Since $S = \sum_{l \in \mathcal{M}}U^*_l + \sum_{{l'} \in \mathcal{N}} W_{l'}$, we have
\begin{align}
\nabla_{kk} S = & -\frac{2a_kM^2(N-1)^2}{N^2\beta^2}-\frac{M(N-1)}{N\beta} \nonumber \\
&  -\sum_{h \in \mathcal{N} \setminus\{k\}}\frac{2a_h M^2}{N^2\beta^2},   \\
  \nabla_{kj} S = & \frac{2(a_k+a_j)M^2(N-1)}{N^2\beta^2}+\frac{M}{N\beta}  -\sum_{h \in \mathcal{N} \setminus\{j,k\}}\frac{2a_h M^2}{N^2\beta^2}.
\end{align}
With {Assumption 1}, it can be derived that
\begin{align}\label{sssd}
        | \nabla_{kk} S |- \sum_{j \in \mathcal{N}\setminus\{k\}} & |\nabla_{kj} S | = |\nabla_{kk} S|- \sum_{j \in \mathcal{N}\setminus\{k\}}  \nabla_{kj} S=0.
\end{align}
Hence, $\mathbf{H}_S$ is symmetric and negative semi-definite, which means $S(\mathbf{p})$ is concave.

\subsection{Proof of Lemma \ref{52}}\label{521}

Note that the expression of $\widetilde{S}(w_1)$ can be derived from quadratic function $S(\mathbf{w})$ by substituting all $w_k$ by $f_k(w_1)$ with (\ref{41}), $k \in \widetilde{\mathcal{N}}$. Since $f_k(w_1)$ is an affine function of $w_1$, then $\widetilde{S}(w_1)$ is quadratic and concave at $w_1$ and the optimal solution $w^*_1$ to {Problem (P4)} can be obtained by solving the following first-order optimality condition
\begin{align}\label{33}
  \nabla \widetilde{S}(w_1) = &  \nabla_1 {S}(\mathbf{w}) + \sum_{k \in \widetilde{\mathcal{N}}} (\nabla_k {S}(\mathbf{w}) \nabla f_k(w_1) )  \nonumber \\
  =  & \nabla_1 {S}(\mathbf{w}) + \sum_{k \in \widetilde{\mathcal{N}}} (\nabla_k {S}(\mathbf{w}) \frac{\mu_k}{\mu_1} )  \nonumber \\
  = & \nabla_1 {S}(\mathbf{w}) + \sum_{k \in \widetilde{\mathcal{N}}} (\nabla_k {S}(\mathbf{w}) \theta_k)   \nonumber \\
   =  & (\nabla {S}(\mathbf{w}))^{\top} \bm{\eta} = 0,
\end{align}
where the first equality complies with the chain rule of composite differentiation, and the second and third equalities are from the definition of $\widetilde{\bm{\mu}}$ in (\ref{fi-1}) and (\ref{fi}). On the other hand, $w_1^*$ is a steady state of (\ref{12ss}) means that there exists certain $\mathbf{w}^*$ such that $\xi(\mathbf{w}^*) = 0$. By recalling (\ref{xxb}), since $\mu_1$ is nonzero, the solution to $\xi(\mathbf{w}) = 0$ is identical to that to (\ref{33}). This completes the proof.

\subsection{Proof of Theorem 2}\label{2131}

Inspired by \cite{x7}, we define a fixed point contraction mapping $\mathbf{\Gamma} (\tau) = (\Gamma_1 (\tau),...,\Gamma_N (\tau))^{\top}:[0,1] \rightarrow \mathbb{R}^N $, where
\begin{equation}\label{xc}
    \mathbf{\Gamma} (\tau)=\tau \mathbf{w}^t + (1-\tau) \mathbf{w}^* + \bm{\eta}\xi (\tau \mathbf{w}^t + (1-\tau) {\mathbf{w}^*}),
\end{equation}
with $\mathbf{w}^*=(w_1^*,(\widetilde{\mathbf{w}}^*)^{\top})^{\top}$. Define $\theta_1=1$, then (\ref{xc}) is equivalent to
\begin{equation}\label{xc1}
    \Gamma_l (\tau)=\tau w^t_l + (1-\tau) w^*_l + \theta_l \xi (\tau \mathbf{w}^t + (1-\tau) {\mathbf{w}^*}), l \in \mathcal{N}.
\end{equation}
Then
\begin{align}\label{edr}
 | \theta_l|  | w^{t+1}_1 -w^*_1 | & = | w^{t+1}_l  -w^*_l |   =  \left| \Gamma_l (1)-\Gamma_l (0) \right| \nonumber \\
 &   =   \left| \int^1_0 \frac{\mathrm{d} \Gamma_l (\tau)}{\mathrm{d} \tau}\mathrm{d}\tau \right|
\leq  \int^1_0 \left|\frac{\mathrm{d}\Gamma_l (\tau)}{\mathrm{d} \tau} \right| \mathrm{d} \tau       \nonumber \\
&  \leq \max\limits_{\tau \in [0,1]} \left|\frac{\mathrm{d}\Gamma_l (\tau)}{\mathrm{d} \tau} \right|,
\end{align}
where we use (\ref{tr1}) and $\xi(\mathbf{w}^*) = 0$. Then, we have
\begin{align}\label{fa1}
       \left| \frac{\mathrm{d} \Gamma_l (\tau)}{\mathrm{d} \tau} \right| = & | (w^t_l-w_l^*) + \mu_1 \theta_l \sum_{l' \in \mathcal{N}} ( \nabla_{1l'} {S}(\mathbf{v}^t(\tau))  \nonumber \\
       &  + \sum_{k' \in \widetilde{\mathcal{N}}} \theta_{k'} \nabla_{k'l'} {S}(\mathbf{v}^t(\tau)) ) (w^t_{l'}-w_{l'}^*) |  \nonumber \\
       = & | (w^t_l-w_l^*) + \mu_1 \theta_l \sum_{l' \in \mathcal{N}} ( \nabla_{1l'} {S} \nonumber \\
       & + \sum_{k' \in \widetilde{\mathcal{N}}} \theta_{k'} \nabla_{k'l'} {S} ) (w^t_{l'}-w_{l'}^*) |  \nonumber \\
       = & |\theta_l| | 1+ \mu_1 \sum_{l' \in \mathcal{N}} \theta_{l'} ( \nabla_{1l'} {S} \nonumber \\
& + \sum_{k' \in \widetilde{\mathcal{N}}} \theta_{k'} \nabla_{k'l'} {S} )  |  |w_1^t - w^*_1| \nonumber \\
       = & |\theta_l || 1 + \mu_1  \bm{\eta}^{\top} \mathbf{H}_S \bm{\eta} |  |w_1^t - w^*_1| \nonumber \\
       = & \Lambda |\theta_l|  |w_1^t - w^*_1| ,
\end{align}
where $\mathbf{v}^t (\tau)= \tau \mathbf{w}^t + (1-\tau) {\mathbf{w}^*}$, $\tau \in [0,1]$, $l \in \mathcal{N}$. The second equality in (\ref{fa1}) holds since the second-order (cross) derivatives of $S$ are constant (see the proof of {Lemma \ref{ov}}). The third equality holds with (\ref{tr1}).

By (\ref{edr}), (\ref{fa1}) and $\bm{\eta}^{\top} \mathbf{H}_S \bm{\eta}\neq 0$, a sufficient condition of the convergence of $w_l$ is
$| w_l^{t+1}-w_l^*| < | w_l^{t}-w_l^*|$, which means $| 1 + \mu_1  \bm{\eta}^{\top} \mathbf{H}_S \bm{\eta} | < 1$. Hence, $\mu_1$ can be chosen in the range determined by
\begin{align}\label{dde}
\mu_1  \bm{\eta}^{\top} \mathbf{H}_S \bm{\eta} \in (-2,0).
\end{align}
Based on (\ref{edr})-(\ref{dde}), a linear convergence rate of $\mathbf{w}^t$ can be obtained, since
\begin{align}\label{}
 |w^t_1  - w^*_1 | & < \Lambda^t |w^0_1 - w^*_1 |, \label{sa1} \\
 |w^t_k  - w^*_k | & = |\theta_k| |w^t_1 - w^*_1 | < \Lambda^t |\theta_k| |w^0_1 - w^*_1 | \nonumber \\
& = \Lambda^t  |w^0_k - w^*_k |, \quad k \in \widetilde{\mathcal{N}}. \label{sa2}
\end{align}
This completes the proof.

\subsection{Proof of Lemma \ref{p1}}\label{pp1}

Note that
\begin{align}\label{rrp}
\bar{\xi}(\mathbf{w}^t) = & (\nabla S^t)^{\top} \bar{\bm{\eta}}^t \mu^t_1 = \nabla_1 S^t \mu^t_1+ \frac{\sum_{k \in \widetilde{\mathcal{N}}}(\nabla_k S^t)^2 \mu^t_1}{ \nabla_1 S^t }.
\end{align}
{{(Prove by Contradiction)}} {{(i)}} Under the precondition $\bar{\xi}(\mathbf{w}^t) \rightarrow 0$, seen from (\ref{rrp}), if $\exists \kappa_1 ,\mu^t_1>0$ such that $\mid \nabla_1 S^t \mid \geq \kappa_1$, then
\begin{align}\label{}
| \bar{\xi}(\mathbf{w}^t) | = & \mid \nabla_1 S^t \mid \mu^t_1 + \frac{\sum_{k \in \widetilde{\mathcal{N}}}(\nabla_k S^t)^2 \mu^t_1}{ \mid \nabla_1 S^t \mid} \nonumber \\
\geq & \kappa_1 \mu^t_1 + \frac{\sum_{k \in \widetilde{\mathcal{N}}}(\nabla_k S^t)^2 \mu^t_1}{ \mid \nabla_1 S^t \mid} \nonumber \\
\geq & \kappa_1 \mu^t_1,
\end{align}
 which contradicts $\bar{\xi}(\mathbf{w}^t) \rightarrow 0$. Therefore, $\nabla_1 S^t \rightarrow 0 $.

 {{(ii)}} If $\exists \kappa_2,\mu^t_1 >0$ such that $\mid \nabla_{k'} S^t \mid \geq \kappa_2$ for certain $k' \in \widetilde{\mathcal{N}}$, then
 \begin{align}\label{}
|\bar{\xi}(\mathbf{w}^t)| = & | \nabla_1 S^t | \mu^t_1 + \frac{\sum_{k \in \widetilde{\mathcal{N}}}( \nabla_k S^t )^2 \mu^t_1}{| \nabla_1 S^t|} \nonumber \\
\geq & | \nabla_1 S^t | \mu^t_1+ \frac{\kappa^2_2 \mu^t_1}{| \nabla_1 S^t| } \nonumber \\
\geq & \frac{\kappa^2_2 \mu^t_1}{| \nabla_1 S^t |} \rightarrow +\infty
 \end{align}
with $\nabla_1 S^t \rightarrow 0 $ (by {{(i)}}), which contradicts $\bar{\xi}(\mathbf{w}^t) \rightarrow 0$. Therefore, $\nabla_{k'} S^t  \rightarrow 0$, $\forall k' \in \widetilde{\mathcal{N}} $.

By combining {{(i)}} and {{(ii)}}, the proof is completed.

\subsection{Proof of Theorem 3}\label{21331}

By (\ref{xcx+2}), we can have
\begin{align}\label{yy}
\mathbf{w}^{t+1}  = & \mathbf{w}^{t} + \bm{\bar{\eta}}^t \bar{\xi}(\mathbf{w}^{t})
 \nonumber \\
  = & \mathbf{w}^{t} +  \mu^t_1 \bm{\bar{\eta}}^t (\nabla S^t)^{\top} \bar{\bm{\eta}}^t \nonumber \\
 = & \mathbf{w}^{t} + \mu_1^t (1,\frac{\nabla_2 S^t}{\nabla_1 S^t},\cdots,\frac{\nabla_N S^t}{\nabla_1 S^t})^{\top} \nonumber \\
 & \cdot ( \nabla_1 S^t + \frac{(\nabla_2 S^t)^2}{\nabla_1 S^t} + \cdots + \frac{(\nabla_N S^t)^2}{\nabla_1 S^t}  ) \nonumber  \\
 = & \mathbf{w}^{t} + \mu_1^t ( 1+ (\frac{\nabla_2 S^t}{\nabla_1 S^t})^2+ \cdots + (\frac{\nabla_N S^t}{\nabla_1 S^t})^2  ) \nabla S^t \nonumber  \\
 = & \mathbf{w}^{t} +  \mu_1^t \parallel \bar{\bm{\eta}}^t \parallel^2 \nabla S^t.
\end{align}
By {Lemma \ref{ov}}, $S$ is concave. To guarantee the convergence of (\ref{yy}), we choose a positive $\mu_1^t$ such that $\mu_1^t \parallel \bar{\bm{\eta}}^t \parallel^2 = 1/L$. Then, by \cite[Thm. 3.3]{bubeck2015convex}, (\ref{yy}) is convergent. Note that the convergence of (\ref{yy}) implies that $\nabla S^t \rightarrow \mathbf{0}$, which is the first-order optimality condition of {Problem (P3)}. Hence, {Theorem \ref{2133}} is proved.

\bibliographystyle{IEEEtran}
\bibliography{myref}

\end{document}